\documentclass[runningheads]{llncs}
\usepackage[T1]{fontenc}
\usepackage[utf8]{inputenc}
\usepackage{graphicx}
\graphicspath{{./graphics/}}

\usepackage[textsize=footnotesize]{todonotes}

\usepackage{scrextend}          %
\usepackage{needspace}          %
\usepackage[ruled,linesnumbered,noend]{algorithm2e}
\usepackage{amsmath,amssymb}
\usepackage{hyperref}
\usepackage[capitalize,noabbrev]{cleveref}
\usepackage{subcaption}
\usepackage{thmtools,cite}
\usepackage{thm-restate}

\declaretheorem[name=Reduction Rule]{redrule}
\declaretheorem[name=Branching Rule]{branchrule}
\declaretheorem[name=Lemma]{lem}

\newcommand{\revision}[1]{\textcolor{black}{#1}}

\newcommand{\prob}[6]{%
  \needspace{3\baselineskip}
  \begin{quote}
    \begin{labeling}{#6}%
    \item[#1]
    \item[\emph{#2}]#3
    \item[\emph{#4}]#5
    \end{labeling}%
  \end{quote}%
}

\makeatletter
\renewcommand{\@Opargbegintheorem}[4]{%
  #4\trivlist\item[\hskip\labelsep{#3#2\@thmcounterend}]}
\makeatother

\newcommand{\probdef}[3]{\prob{#1}{Instance:}{#2}{Question:}{#3}{as}}

\newcommand{\proofparagraph}[1]{\smallskip\emph{#1}}

\SetKwFunction{FMain}{Recursive-\revision{\bcrextension}}
\SetKw{False}{\textsf{\textup{false}}}
\DeclareMathOperator*{\argmax}{arg\,max}
\usepackage[basic]{complexity}

\newcommand{\blocks}{\ensuremath{\mathsf{blocks}}}
\newcommand{\breakpoints}{\ensuremath{\mathsf{bp}}}
\newcommand{\id}{\ensuremath{\mathsf{id}}}
\newcommand{\N}{\ensuremath{\mathbb{N}}}
\newcommand{\leaves}{\ensuremath{\mathsf{leaf}}}
\newcommand{\rroot}{\ensuremath{\mathsf{root}}}
\newcommand{\lc}{\ensuremath{\mathsf{lc}}}
\newcommand{\rc}{\ensuremath{\mathsf{rc}}}
\newcommand{\child}{\ensuremath{\mathsf{ch}}}
\newcommand{\parent}{\ensuremath{\mathsf{par}}}
\newcommand{\ances}{\ensuremath{\mathsf{anc}}}
\newcommand{\depth}{\ensuremath{\mathsf{depth}}}
\newcommand{\opt}{\ensuremath{\mathsf{opt}}}
\newcommand{\rob}{\ensuremath{\mathsf{rob}}}
\newcommand{\parr}{\ensuremath{\mathsf{par}}}
\newcommand{\preced}{\ensuremath{\mathsf{prec}}}
\newcommand{\lca}{\ensuremath{\mathsf{lca}}}
\newcommand{\inv}{\ensuremath{\mathsf{inv}}}
\newcommand{\gl}{\ensuremath{\mathsf{gl}}}
\newcommand{\adjbef}{\ensuremath{\mathsf{adjbef}}}
\newcommand{\adjaft}{\ensuremath{\mathsf{adjaft}}}
\newcommand{\bcrproblong}{\textup{\textsc{One-Tree Block Crossing Minimization}}}
\newcommand{\bcrprob}{\textup{\textsc{OTBCM}}}
\newcommand{\bcrextension}{\textup{\textsc{OTBCM-Extension}}}

\newcommand{\transsort}{\textup{\textsc{Sorting by Transpositions}}}

\renewcommand{\orcidID}[1]{\href{https://orcid.org/#1}{\includegraphics[scale=.03]{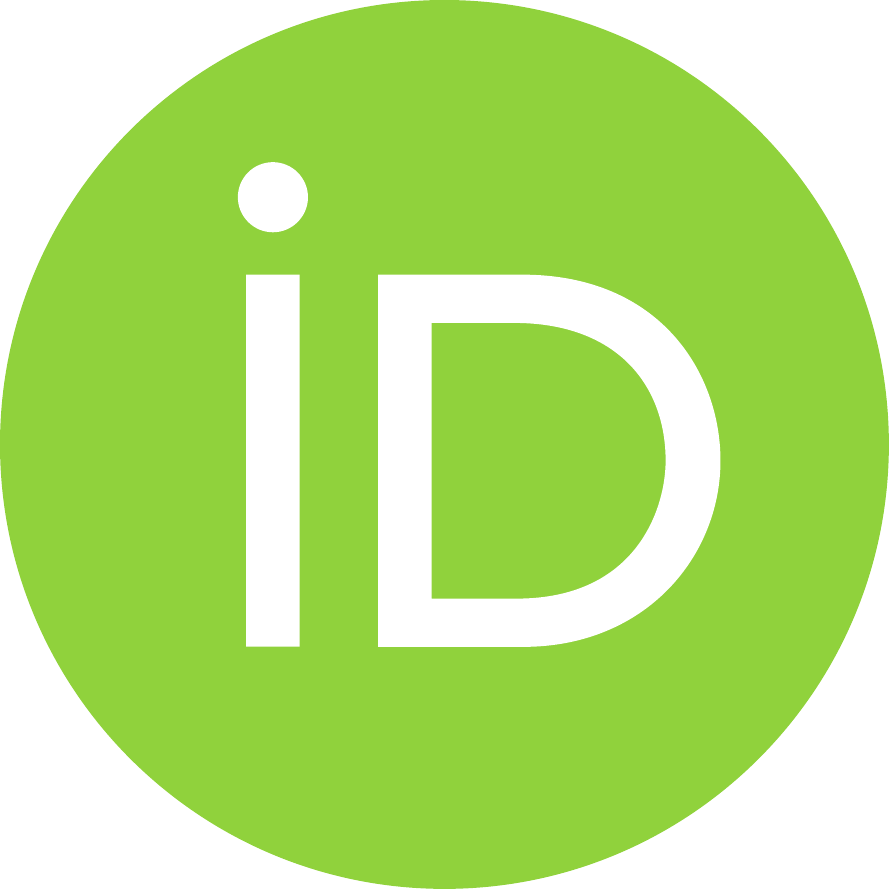}}} 

\begin{document}
\title{Block Crossings in One-Sided Tanglegrams\texorpdfstring{\thanks{This work has been supported by the Vienna Science and Technology Fund (WWTF) [10.47379/ICT19035].}}{}}
\author{Alexander Dobler \orcidID{0000-0002-0712-9726}
\and
  Martin Nöllenburg \orcidID{0000-0003-0454-3937}
}
\authorrunning{A. Dobler and M. Nöllenburg}
\institute{Algorithms and Complexity Group, TU Wien, Vienna, Austria \email{\{adobler,noellenburg\}@ac.tuwien.ac.at}}
\maketitle              %
\begin{abstract}
    Tanglegrams are drawings of two rooted binary phylogenetic trees and a matching between their leaf sets. The trees are drawn crossing-free on opposite sides with their leaf sets facing each other on two vertical lines. Instead of minimizing the number of pairwise edge crossings, we consider the problem of minimizing the number of \emph{block crossings}, that is, two bundles of lines crossing each other locally.

    \looseness=-1
    With one tree fixed, the leaves of the second tree can be permuted according to its tree structure. We give a complete picture of the algorithmic complexity of minimizing block crossings in one-sided tanglegrams by showing \NP-completeness, constant-factor approximations, and a fixed-parameter algorithm. We also state first results for non-binary trees.
\end{abstract}

\section{Introduction}
\looseness=-1
Tanglegrams~\cite{pageTangledTreesPhylogeny2003} are drawings of two rooted $n$-leaf trees and a matching between their leaf sets drawn as straight edges. The trees are drawn such that one tree is on the left and the other is on the right with their leaf sets facing each other on two vertical lines (see \cref{fig:tlp}).
An important application of tanglegrams is the comparison of two phylogenetic trees with the same leaf set~\cite{szh-trptn-11,vajg-utcttd-10}, which can be used to study co-speciation or for comparison of hypothetical phylogenetic trees computed by different algorithms. Other applications are comparisons of dendrograms in hierarchical clustering or software hierarchies~\cite{hw-vchod-08}. The readability of tanglegrams heavily depends on the order of the two leaf sets on the vertical lines, as this determines the number of edge crossings between the matching edges. The possible orders depend on the tree structure of both trees, so finding appropriate orders that minimize the number of pairwise crossings is a nontrivial problem known as the \textsc{Tanglegram Layout Problem} (TLP)~\cite{fernauComparingTreesCrossing2010,ds-olohdptv-04,buchinDrawingCompleteBinary2012}.
In this paper we focus on minimizing block crossings instead of pairwise edge crossings; see \cref{fig:tlpblockcrossing} for an example. 
This is done by relaxing that the matching edges must be drawn as straight lines, but rather drawing them as $x$-monotone curves, which allows for shifting and grouping crossings more flexibly. 
A \emph{block crossing}~\cite{DBLP:journals/jgaa/DijkFFLMRSW17} is then defined as a crossing between two disjoint sets of edges in a confined region $R$, each of which forms a bundle of locally parallel curves in $R$; no further edge may intersect with $R$.
Block crossings provide the ability to group multiple crossings together, instead of having them scattered throughout the drawing. Furthermore, this mostly leads to fewer block crossings than the required number of pairwise crossings and thus reduces visual clutter~\cite{fhsv-bceg-16}. We initiate the work on block crossings for tanglegrams and focus on tanglegrams for binary trees where the leaf order of one tree is fixed.
\begin{figure}[t]
    \centering
    \begin{subfigure}[b]{0.48\textwidth}
        \centering
        \includegraphics[width=.8\linewidth]{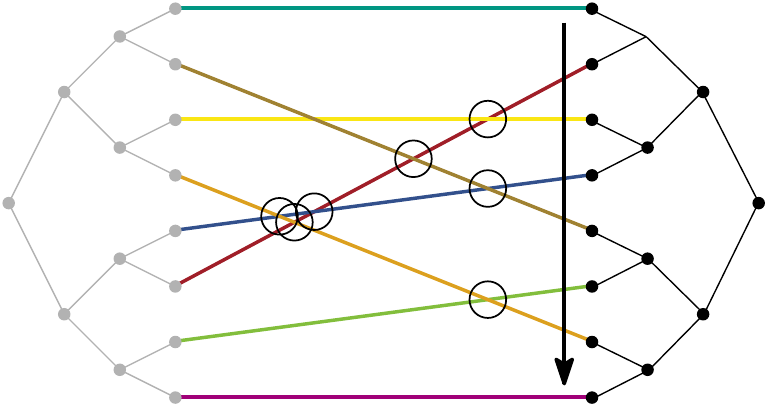}
        \caption{Optimal solution with 7 pairwise crossings.}
        \label{fig:tlp}
    \end{subfigure}
    \hfill
    \begin{subfigure}[b]{0.48\textwidth}
        \centering
        \includegraphics[width=.8\linewidth]{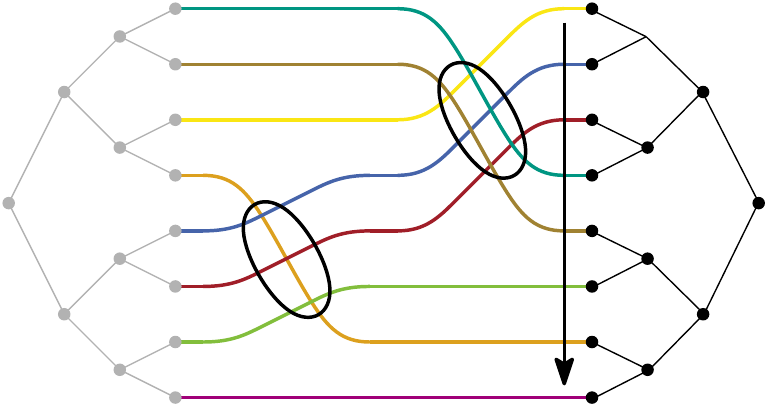}
        \caption{Optimal solution with 2 block crossings (but 9 pairwise crossings).}
        \label{fig:blockcrossing}
    \end{subfigure}
    \caption{Two tanglegrams for the same input. The leaf order of the left tree is fixed, the right tree can be permuted to minimize pairwise (a) or block crossings~(b).}
    \label{fig:tlpblockcrossing}
\end{figure}
\paragraph*{Related work.}
The TLP has been studied mainly for binary trees, as phylogenetic trees are mostly binary: If both leaf sets can be ordered then the problem is known to be \NP-complete \cite{fernauComparingTreesCrossing2010}, even if the trees are complete \cite{buchinDrawingCompleteBinary2012}. However, there exist approximation results \cite{buchinDrawingCompleteBinary2012}, fixed-parameter algorithms \cite{fernauComparingTreesCrossing2010,buchinDrawingCompleteBinary2012,bockerFasterFixedParameterApproach2009}, integer linear programming formulations \cite{DBLP:conf/wea/BaumannBL10}, and heuristics \cite{nollenburgDrawingBinaryTanglegrams2009}. If the leaf order of one tree is fixed, then the problem is solvable in polynomial time \cite{fernauComparingTreesCrossing2010}; if the trees are not binary, however, then even this problem is \NP-complete \cite{DBLP:conf/cpm/BulteauGS22}.

\emph{Edge bundling} is a technique in network visualization that groups multiple edges together to reduce visual clutter. If two edge bundles locally cross each other this is called a \emph{bundled crossing}~\cite{fhsv-bceg-16,afp-bcn-16}. For a collection of results on bundled crossings in general graphs we refer to \cite{DBLP:books/sp/20/Nollenburg20}.
Essentially, a \emph{block crossing} is the same as a bundled crossing.
But the term block crossing has been used mainly in contexts, where bundled crossings are determined in a purely combinatorial way using permutations, and no topology is required. Our work is similar, so we adopt that term, and give an overview of previous work on block crossings: Fink et al.~\cite{DBLP:journals/jgaa/FinkPW15} worked on minimizing block crossings amongst metro lines in a pre-specified metro network where multiple metro lines might be routed along the same edge. They considered general metro networks but mostly focused on special cases such as paths, trees, and upward trees. Their results mostly include approximation and fixed-parameter algorithms, as even the most restricted variant of their problem is \NP-complete.
Van Dijk et al.~\cite{DBLP:journals/jgaa/DijkFFLMRSW17} studied block crossings in the context of storylines. They showed \NP-completeness, fixed-parameter tractability, and an approximation algorithm.
Both of the above works on block crossings pointed out the connection between block crossings and a problem stemming from computational biology called \transsort~\cite{DBLP:journals/siamdm/BafnaP98}, where a permutation has to be transformed into the identity permutation by exchanging adjacent blocks of elements, calling this operation a \emph{transposition}.

\paragraph*{Our contribution.}
\looseness=-1
We study block crossings in the context of tanglegrams. More precisely, we are given two rooted $n$-leaf trees $T_1$ and $T_2$ and a matching between their leaf sets. The order of the leaves in $T_1$ is fixed. Our problem is to find a permutation of the leaves of $T_2$ that is consistent with the tree structure, and admits the minimal number of block crossings routing the matching edges from $T_2$ to $T_1$.
For a formal problem definition we refer to \cref{sec:bcr}. 
We focus mostly on binary trees $T_2$: In \cref{sec:hardness} we show that the problem is \NP-complete even for complete binary trees. In \cref{sec:approx} we give 2.25-approximation algorithms, the first for general binary trees, and a  faster second one for complete binary trees. In \cref{sec:fpt} we show that the problem is fixed-parameter tractable (FPT) parameterized by the number of block crossings.
In \cref{sec:beyond} we show that the techniques used in \cref{sec:approx} to find a polynomial-time approximation algorithm do not extend to non-binary trees.
We start by giving some preliminaries in the following section.

\section{Preliminaries}
Let $\delta_{i,j}$ be the Kronecker delta function that is 1 if $i=j$ and 0 otherwise.
Let $[n]=\{1,\dots, n\}$.
\paragraph*{Permutations.}
A \emph{permutation} $\pi:[k]\to X$ is a bijective function from~$[k]$ to a set~$X$, in particular, $\pi:[k]\to [k]$ is a permutation of $[k]$, and $\inv(\pi)$ is its inverse.
We write $\pi_i$ for $\pi(i)$ and use superscript if we want to tell apart multiple permutations.
Sometimes we write permutations~$\pi$ as sequences of elements $(\pi_1,\dots,\pi_k)$.
We denote by~$\Pi_n$ the set of all permutations from~$[n]$ to~$[n]$.
For a permutation $\pi\in\Pi_n$ and $i\in [n]$, let $\pi\ominus\pi_i$ be the permutation of $[n-1]$ obtained by first removing $\pi_i$ from $\pi$, and then decreasing all elements of $\pi$ greater than $\pi_i$ by one.
For two permutations $\pi=(x_1,\dots,x_n)$ and $\pi^\prime=(y_1,\dots,y_m)$ of disjoint elements, we denote by $\pi\star \pi^\prime$ their \emph{concatenation} $(x_1,\dots x_n,y_1,\dots,y_m)$.
For two sets~$\Pi$ and~$\Pi^\prime$ of permutations, we define $\Pi\star \Pi^\prime=\{\pi\star \pi^\prime\mid \pi\in\Pi,\pi^\prime\in\Pi^\prime\}$.
\paragraph*{Transpositions.}\label{par:transpositions}
A \emph{transposition} $\tau=\tau(i,j,k)\in \Pi_n$ with $1\le i<j<k\le n+1$ is the permutation
\[(1,\dots, i-1,j,\dots, k-1, i,i+1,\dots, j-2,j-1,k,\dots, n).\]
This definition is different from the classic transpositions in discrete mathematics as it stems from computational biology \cite{DBLP:journals/siamdm/BafnaP98}.
Assume $\pi\in \Pi_n$.
The permutation $\pi\circ \tau(i,j,k)$ has the effect of swapping the contiguous subsequences $\pi_i,\dots,\pi_{j-1}$ and $\pi_j,\dots, \pi_{k-1}$.
A \emph{block} in a permutation $\pi\in \Pi_n$ is a maximal contiguous subsequence $\pi_i,\pi_{i+1},\dots,\pi_j$ that is also a contiguous subsequence of the identity permutation $\id_n$. The number of blocks in $\pi$ is denoted by $\blocks(\pi)$.
An index $i\in [n]\cup\{0\}$ is a \emph{breakpoint} if
\begin{itemize}
    \item $i=0$ and $\pi_1\ne 1$,
    \item $1\le i\le n-1$ and $\pi_i+1\ne \pi_{i+1}$, or
    \item $i=n$ and $\pi_n\ne n$.
\end{itemize}
Essentially, a breakpoint in $\pi$ corresponds to a pair $(x,y)$ of adjacent elements in the extended permutation $\pi^e=(0)\star \pi\star (n+1)$ such that $x+1\ne y$.
Each breakpoint $i$ in $\pi$ has a corresponding breakpoint element $\pi^e_{i+1}$ (this can include $0$). Conversely, we say that $\pi^e_{i+1}$ corresponds to breakpoint $i$. Let $\breakpoints(\pi)$ be the number of breakpoints in $\pi$.
The transposition distance $d_t(\pi)$ of $\pi$ is the minimum number $k\in \N_0$ such that there exist transpositions $\tau^1,\dots,\tau^k$ with $\pi\circ \tau^1\circ\dots \circ\tau^k=\id_n$. In this case we call $\tau^1,\dots,\tau^k$ an $\id$-transposition sequence for $\pi$. \revision{Note that there always exists an $\id$-transposition sequence, as every permutation can be transformed to the identity-permutation by adjacent swaps (cf.\ Bubblesort).}
Let $\pi$ be a permutation with $r>0$ breakpoints. Then $\gl(\pi)\in\Pi_{r-1}$ is formed by ``gluing'' each block together into a single element. Furthermore, if $\pi$ starts with 1 then the first block is removed, and if $\pi$ ends with $n$ then the block at the end is removed (see \cite{christieGenomeRearrangementProblems1998}).
For instance, if $\pi=(3,1,2,8,9,4,5,6,7,10)$, then $\gl(\pi)=(2,1,4,3)$.
Two important lemmata that will be used throughout the paper are given below.
\begin{lem}[\hspace{-0.01pt}\cite{DBLP:journals/siamdm/BafnaP98}]\label{lemma:lowerbound}
    For $\pi\in\Pi_n$ we have $d_t(\pi)\ge \lceil\frac{\blocks(\pi)-1}{3}\rceil$ and $d_t(\pi)\ge \lceil\frac{\breakpoints(\pi)}{3}\rceil$.
\end{lem}
\begin{lem}[\hspace{-0.01pt}\cite{christieGenomeRearrangementProblems1998}]\label{lemma:gluetransdist}
    For $\pi\in\Pi_n$, $d_t(\pi)=d_t(\gl(\pi))$.
\end{lem}
A well-studied problem in genome rearrangement is \transsort. It asks for a permutation $\pi$ and an integer $k$, whether $d_t(\pi)\le k$.
It is known that this problem is NP-complete \cite{bulteauSortingTranspositionsDifficult2012}, and the authors even showed the following result which will be used in our paper.
\begin{lem}[\hspace{-0.01pt}\cite{bulteauSortingTranspositionsDifficult2012}]\label{lemma:hardnesslowerbound}
    For $\pi\in\Pi_n$ it is \NP-hard to decide whether $d_t(\pi)=\frac{\breakpoints(\pi)}{3}$.
\end{lem}
But there is a simple fixed-parameter algorithm outlined by Mahajan et al.~\cite{mahajanApproximateBlockSorting2006}:
First, if $\breakpoints(\pi)>3k$, we can immediately report that $(\pi,k)$ is a no-instance by \cref{lemma:lowerbound}.
Otherwise, we search for an $\id$-transposition sequence of length at most $k$ for $\gl(\pi)$ (see \cref{lemma:gluetransdist}) using a simple search tree approach. As $\gl(\pi)$ has at most $3k$ elements, there are only $\mathcal{O}((3k)^3)$ possible transpositions. The search-tree depth is at most $k$, as we can perform at most $k$ transpositions. Thus, we can determine in time $\mathcal{O}(n(3k)^{3k})$ if an $\id$-transposition sequence for $\gl(\pi)$ exists, and also report it in the positive case.
This transposition sequence can be easily transformed into an $\id$-transposition sequence for $\pi$ by transposing the blocks of $\pi$ corresponding to the elements in $\gl(\pi)$ for each transposition.
\paragraph*{Trees}
We only consider ordered rooted trees $T$. Let $\rroot(T)$ be the \emph{root} of $T$. Let $\leaves(T)$ be the set of \emph{leaves} of $T$.
For $v\in V(T)$ let $\depth(v)$ be the length of the shortest path between $\rroot(T)$ and $v$ in $T$.
For $v\in V(T)$, let $T(v)$ be the subtree of $T$ rooted at $v$.
For an internal node $v\in V(T)$ let $\child(v)$ be the set of children of $v$ and let $\parent(v)$ be the parent of $v$. Further, let $\ances_T(v)$ be the set of \emph{ancestors} of $v$ in $T$.
For two distinct vertices $v,w\in V(T)$ let $\lca(v,w)$ be the \emph{lowest common ancestor} of $v$ and $w$.
Two vertices $v,w\in V(T)$ are \emph{siblings} if they are children of the same vertex.
If $T$ is a binary tree, then we denote the two children by $\lc(v)$ and $\rc(v)$.

A rooted tree $T$ encodes a set of permutations $\Pi(T)$ of its leaves which can be obtained by permuting children of an inner node: Namely, let $\Pi(T(v))=\revision{\{(v)\}}$ if $v\in \leaves(T)$. If $v\not\in \leaves(T)$, let $\child(v)=\{w_1,\dots,w_k\}$ and we define
\[\Pi(T(v))=\bigcup_{\psi\in \Pi_{k}}\Pi(T(w_{\psi_1}))\star\Pi(T(w_{\psi_2}))\star\dots\star\Pi(T(w_{\psi_k})).\]
If $\pi\in\Pi(T)$, we say that $\pi$ is \emph{consistent} with $T$.

\section{Block Crossings in Tanglegrams}\label{sec:bcr}
In this section we want to properly define the problem we are dealing with.
To reiterate, we are given two trees $T_1$ and $T_2$ with $n$ leaves, a matching between their leaf sets, and a fixed leaf order of $T_1$. Notice, that w.l.o.g.\ we can assume that the leaf sets of both trees are labelled with $[n]$, and that the matching edges are between leaves labelled with the same integer. Further, by relabelling we can assume that the leaf order of $T_1$ is the identity permutation $\id_n$.
Now we want to find (1) a permutation $\pi$ of the leaves of $T_2$ that conforms to the structure of $T_2$, and (2) a sequence of block crossings that route the matching edges from $T_1$ to $T_2$. (1) means that we are looking for a permutation $\pi\in \Pi(T_2)$. As a block crossing only changes the vertical order of two blocks of matching edges, (2) asks for a sequence of swaps of adjacent blocks of edges---which is a purely combinatorial procedure. Further, notice that each block crossing can be modelled by a transposition $\tau$ on the vertical order of edges. Hence, instead of looking for a sequence of block crossings, we equivalently look for a sequence of transposition transforming $\pi$ into the identity permutation. This leads to the following decision variant of our problem, where $T:=T_2$.
\probdef{\revision{\bcrproblong}~(\revision{\bcrprob})}{A rooted tree $T$ with $\leaves(T)=[n]$ and a positive integer $k$.}{Is there a permutation $\pi\in\Pi(T)$ such that $d_t(\pi)\le k$?}
Our algorithms will produce a witness in case of a YES-instance, that is, a permutation $\pi\in \Pi(T)$, and a $\id$-transposition sequence $\tau^1,\dots,\tau^\ell$ of $\pi$ with $\ell\le k$. In the following sections we will investigate the algorithmic complexity of this problem. We start with results that assume that $T$ is a binary tree.
\section{NP-Hardness}\label{sec:hardness}
\revision{\bcrprob}\ implicitly contains as a subproblem to sort a permutation by a sequence of few transpositions. As \transsort\ is \NP-complete \cite{bulteauSortingTranspositionsDifficult2012}, this suggests that \revision{\bcrprob}\ is also \NP-complete. We show this for the restricted case, where the input tree is complete and binary. The proof, however, is not as straight-forward as the relation between \revision{\bcrprob}\ and \transsort\ might suggest. 
The main idea is to construct for an input permutation $\pi$ a tree $T$ such that there exists $\pi'\in \Pi(T)$ with $\gl(\pi)=\gl(\pi')$, and $\pi'$ is the only permutation in $\Pi(T)$ that minimizes the number breakpoints. As it is already \NP-hard to decide whether $d_t(\pi)$ equals the lower bound $\breakpoints(\pi)/3$ (\cref{lemma:hardnesslowerbound}), this shows \NP-hardness for \revision{\bcrprob}.
\begin{restatable}{theorem}{thmnpcomp}\label{thm:NPcomp}
    \revision{\bcrprob}\ is \NP-complete for complete binary trees.
\end{restatable}

Let $\pi$ be an arbitrary permutation consisting of $2^p-1$ elements where $p\in\mathbb{N}$. We define a labelled complete binary tree $T$ on $2^{p+1}$ leaves. We label the leaves of $T$ such that the leftmost leaf is labelled $1$, and the rightmost leaf is labelled $2^{p+1}$. The remaining leaves are labelled from left to right by
\[2\pi_1,2\pi_1+1, 2\pi_2,2\pi_2+1,\dots,2\pi_{2^p-2},2\pi_{2^p-2}+1,2\pi_{2^p-1},2\pi_{2^p-1}+1.\]
The leaves of $T$ define the permutation $\pi^I=\pi^I(T)$ as follows. 
We have $\pi^I(T)=\pi^I(\rroot(T))$ where
\[\pi^I(v)=\begin{cases}(v),                             & \text{if }v\text{ is a leaf}, \\
             \pi^I(\lc(v))\star\pi^I(\rc(v)), & \text{otherwise.}
\end{cases}\]
Notice that $\breakpoints(\pi^I)=\breakpoints(\pi)$. Refer to \cref{fig:nphardexample} for an example of $T$.
\begin{figure}[t]
    \centering
    \includegraphics{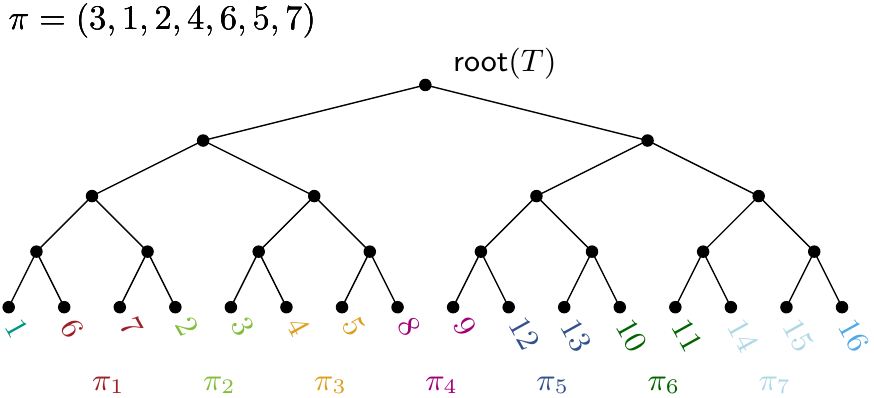}
    \caption{An example for $T$ in the proof of Theorem~\ref{thm:NPcomp} depicting the permutation $\pi=(3,1,2,4,6,5,7)$.}
    \label{fig:nphardexample}
\end{figure}

Next we show that $\pi^I$ has the minimum number of breakpoints out of all the permutations in $\Pi(T)$.
\begin{lem}\label{claim:manybreakpoints}
    For all $\pi^\prime\in \Pi(T)$ such that $\pi^\prime\ne \pi^I$ we have that $\breakpoints(\pi^\prime)>\breakpoints(\pi^I)$.
\end{lem}
\begin{proof}
    Let $\pi^\prime\in \Pi(T)$ with $\pi^\prime\ne \pi^I$ be arbitrary.
    Let $X=\{x_1,\dots,x_\ell\}$ be the set of breakpoint elements in $\pi^I$ that do not appear in $\pi^\prime$. Furthermore, let $Y=\{y_1,\dots, y_m\}$ be the set of breakpoint elements that appear in $\pi^\prime$ but not appear in $\pi^I$. We will prove that $m>\ell$ by showing that each breakpoint element that is present in $X$ attributes to some breakpoint elements being in $Y$.

    First notice that all elements $x\in X$ are odd and are the left child of its parent.
    As $x$ is not a breakpoint element in $\pi^\prime$ it has to appear before $x+1$ in $\pi^\prime$, where $x+1$ is even and the right child of its parent.
    Thus, in $\pi^\prime$ the parent of $x$ and $x+1$ has its children appear in reversed order.
    This means that $x-1$ and $x+1$ are breakpoint elements in $Y$ (for the case of $x=1$, recall that 0 can also be a breakpoint element).
    For an example see \cref{fig:nphardexample}: If $9\in X$ then it has to appear before $10$ in $\pi^\prime$. Thus, the parents of 9 and 10 have to be flipped, resulting in 8 and 10 being breakpoint elements in $\pi^\prime$.

    Now we have that for any set of elements $X$ being all odd, $Y\supseteq \{x-1,x+1\mid x\in X\}$. It follows that $|X|<|Y|$ because if $|X|=1$ then $|Y|=2$ and otherwise $|Y|\ge |X|+1$ by induction on the size of $|X|$. This is what we wanted to show. \qed
\end{proof}
\begin{proof}[Proof of Theorem~\ref{thm:NPcomp}]
    \NP-membership is immediate as we are dealing with a purely combinatorial problem with polynomial-size witnesses.
    For a permutation $\pi$ it is \NP-hard to decide whether $d_t(\pi)=\frac{\breakpoints(\pi)}{3}$ \cref{lemma:hardnesslowerbound}.
    We will show that, given an arbitrary permutation $\pi$, we can find an instance $(T,k)$ of \revision{\bcrprob} such that $d_t(\pi)=\frac{\breakpoints(\pi)}{3}$ if and only if $(T,k)$ is a \emph{YES}-instance of \revision{\bcrprob}. This is enough to show that $\revision{\bcrprob}$ is NP-hard.

    Let $\pi$ be an arbitrary permutation. By \cref{lemma:gluetransdist} we may assume that $\pi$ is of length $2^p-1$ for some $p\in\mathbb{N}$: If $\pi$ is not of length $2^p-1$, e.g., of size $2^p-1-\ell$, we simply replace the last element $e$ of $\pi$ by the block $e,e+1,\dots, e+\ell$ and replace $e^\prime$ by $e^\prime+\ell$ for all other $e^\prime>e$. We have that the new permutation $\pi^a$ has the same number of breakpoints, and $d_t(\pi)=d_t(\pi^a)$ as $\gl(\pi)=\gl(\pi^a)$. Thus, deciding $d_t(\pi)=\frac{\breakpoints(\pi)}{3}$ is equivalent to deciding $d_t(\pi^a)=\frac{\breakpoints(\pi)}{3}$. By choosing $p$ appropriately, $\pi^a$ has at most double the size of $\pi$.
    Henceforth, we assume that $\pi$ is of size $2^p-1$.

    Let $T$ be the complete binary tree as described above and let $k=\frac{\breakpoints(\pi)}{3}=\frac{\breakpoints(\pi^I(T))}{3}$. We now show the stated equivalence.
    If $d_t(\pi)=\frac{\breakpoints(\pi)}{3}$ then $(T,k)$ is a \emph{YES}-instance as $d_t(\pi^I(T))=d_t(\pi)=\frac{\breakpoints(\pi)}{3}=k$. The equality $d_t(\pi^I(T))=d_t(\pi)$ holds as $\gl(\pi^I(T))=\gl(\pi)$.

    Conversely, if $d_t(\pi)\ne\frac{\breakpoints(\pi)}{3}$ then $d_t(\pi^I(T))> k$ (the distance cannot be smaller by \cref{lemma:lowerbound}).
    For all other $\pi^\prime\in\Pi(T)$ with $\pi^\prime\ne\pi^I(T)$ we have that $\breakpoints(\pi^\prime)>\breakpoints(\pi^I(T))$ by \cref{claim:manybreakpoints}.
    By \cref{lemma:lowerbound}, $d_t(\pi^\prime)\ge \lceil \frac{\breakpoints(\pi^\prime)}{3}\rceil>\frac{\breakpoints(\pi^I(T))}{3}=k$. Summarizing,  we have that $\forall \pi\in \Pi(T):d_t(\pi)>k$ and $(T,k)$ is a \emph{NO}-instance. \qed
\end{proof}

\section{Approximation}\label{sec:approx}
We present two approximation algorithms for \revision{\bcrprob}\ on binary trees. These rely on a polynomial algorithm by Walter and Dias~\cite{DBLP:conf/spire/WalterDM00} that sorts a permutation $\pi$ of length $n$ by at most $\frac{3}{4}\breakpoints(\pi)$ transpositions in time $\mathcal{O}(n^2)$. Together with the lower bound of \cref{lemma:lowerbound} this implies a 2.25-approximation algorithm for \transsort\footnote{There exists a better approximation algorithm based on breakpoints leading to a 2-approximation, but the authors do not state any runtime \cite{DBLP:journals/dm/ErikssonEKSW01}.}. 
Hence, our algorithms need to find the permutation $\pi\in \Pi(T)$ that minimizes the number of breakpoints. Applying the algorithm of Walter and Dias to this permutation gives a 2.25-approximation for \revision{\bcrprob}. 
We present two algorithms for minimizing the number of breakpoints. In the case of binary trees we present an $\mathcal{O}(n^3)$ algorithm, in the case of complete binary trees we can improve this to $\mathcal{O}(n^2)$. In both algorithms we start by discussing how to minimize the number of blocks instead of breakpoints, then we show how the algorithms can be adapted for breakpoints. Note that the number of blocks and breakpoints differ by at most one for any permutation, hence, minimizing blocks also leads to similar approximations.

\paragraph*{Binary trees.}
We start with an algorithm for binary trees. Essentially, finding $\pi\in \Pi(T)$ with minimal blocks is closely to the following previously studied problem:
Given a complete binary tree $T$ with $n$ leaves and an arbitrary distance function $d:\leaves(T)\times \leaves(T)\to \mathbb{R}$ find $\pi\in \Pi(T)$ that minimizes $\sum_{i=1}^{n-1}d(\pi_i,\pi_{i+1})$. Defining $d(i,j)=(1-\delta_{i+1,j})$ for $i,j\in\leaves(T)$ exactly captures our problem (note that $d$ is not symmetric). Bar-Joseph et al.~\cite{barjosephKaryClusteringOptimal2003} gave an algorithm in the case where $T$ is a binary tree and the distance function $d$ is symmetric. We cannot directly apply their algorithm, so we show an adaptation to our problem, large parts of the algorithm are the same.
The main idea is to do a bottom up dynamic program that computes the optimal ordering for the leaves of a subtree when the leftmost and rightmost leaf is fixed.
\begin{restatable}{theorem}{thmnthirdapprox}\label{thm:nthirdapprox}
    Let $T$ be a rooted binary tree with $\leaves(T)=[n]$.
    Then a permutation $\pi\in\Pi(T)$ minimizing the number of blocks (breakpoints) can be computed in $\mathcal{O}(n^3)$ time and $\mathcal{O}(n^2)$ space.
\end{restatable}
\begin{proof}
    We show how to find the permutation minimizing the number of blocks, and discuss how to minimize the number of breakpoints at the end of the proof.
    We proceed as in \cite{barjosephKaryClusteringOptimal2003} by giving a bottom-up dynamic program.
    We compute for each vertex $v\in V(T)$ and each pair of leaves $i,j\in \leaves(T)$ such that $v=\lca(i,j)$ the value $B(v,i,j)$ that is the minimum number of blocks for a permutation of $\leaves(T(v))$ such that $i$ is the first element and $j$ is the last element of that permutation. This is done by bottom-up dynamic programming.
    Clearly $B(v,v,v)=1$ for $v\in\leaves(T)$.
    Then for all other $v\in V(T)$ such that $w=\lc(v)$, $x=\rc(v)$, $i\in \leaves(T(w))$, and $j\in \leaves(T(x))$,
    \[B(v,i,j)=\min\{B(w,i,h)+B(x,\ell,j)+(1-\delta_{h+1,\ell})\mid h\in \leaves(T(w)),\ell\in\leaves(T(x))\}.\]
    The value $B(v,j,i)$ is computed equivalently by exchanging the roles of $w$ and $x$.
    It is easy to see that
    \[\min_{\pi^\prime\in\Pi(T)}\blocks(\pi^\prime)=\min \{B(\rroot(T),i,j)\mid i,j\in \leaves(T),\lca(i,j)=\rroot(T)\},\]
    and a permutation $\pi$ that achieves this number of blocks can also be obtained by saving $h\in\leaves(T(w))$ and $\ell\in\leaves(T(x))$ at $v$ that minimize $B(v,i,j)$.

    A straight-forward way of computing $B(v,i,j)$ would result in an $\mathcal{O}(n^4)$ algorithm, but there is a faster way, as described by Bar-Joseph et al.~\cite{barjosephKaryClusteringOptimal2003}. To compute $B(v,i,j)$ we save the intermediate values $\text{Temp}(i,\ell)$ which are computed as
    \[\text{Temp}(i,\ell)=\min_{h\in T(w)}B(w,i,h)+(1-\delta_{h+1,\ell}).\]
    Then we can compute $B(v,i,j)$ as
    \[B(v,i,j)=\min_{\ell\in T(x)}\text{Temp}(i,\ell)+B(x,\ell,j).\]
    Both computations take $\mathcal{O}(n)$ time, and $B(v,j,i)$ can be computed similarly.
    For each pair $i,j\in \leaves(T)$, the values $B(v,i,j)$ and $B(v,j,i)$ are only computed for one $v=\lca(i,j)$, and the computation takes $\mathcal{O}(n)$ time.
    Thus, the overall runtime is $\mathcal{O}(n^3)$, and the overall space complexity is $\mathcal{O}(n^2)$.

    The above algorithm can be adapted to minimize the number of breakpoints instead of the number of blocks: After the algorithm is completed we can look at all values $B(\rroot(T), i,j)$. If $i$ is not $1$, we increase the value by one. If $j$ is not $n$, we increase the value by one. From these new $B(\rroot(T), i,j)$-values we can find the minimum number of breakpoints. The permutation $\pi\in\Pi(T)$ achieving this number of breakpoints can be obtained equivalently as above. \qed
\end{proof}

\paragraph*{Complete binary trees.}
If we are dealing with complete binary trees we can give a faster algorithm for minimizing blocks (breakpoints). This algorithm is based on an algorithm by Brandes~\cite{brandesOptimalLeafOrdering2007} for the problem of finding $\pi\in\Pi(T)$ that minimizes $\sum_{i=1}^{n-1}d(\pi_i,\pi_{i+1})$ as described above. The algorithm of Brandes solves the problem for complete binary trees in $\mathcal{O}(n^2\log n)$ time and $\mathcal{O}(n)$ space.
This already gives an algorithm for our problem on complete binary trees $T$ by setting $d(i,j)=(1-\delta_{i+1,j})$ for $i,j\in\leaves(T)$. But as we are not dealing with an arbitrary distance function, we can find an even faster algorithm.

\begin{theorem}\label{thm:nsquard}
    Let $T$ be a complete rooted binary tree with $\leaves(T)=[n]$ s.t.\ $n=2^k$.
    Then a permutation $\pi\in\Pi(T)$ minimizing the number of blocks (breakpoints) can be computed in $\mathcal{O}(n^2)$ time and $\mathcal{O}(n\log n)$ space.
\end{theorem}

We start by explaining the algorithm for blocks. The main idea is that, instead of applying dynamic programming bottom-up, we can apply dynamic programming building the permutation from left to right.

Brandes~\cite{brandesOptimalLeafOrdering2007} pointed out that fixing a leaf of $T$ to be at a specific position of a permutation $\pi\in\Pi(T)$ determines a partition into preceding and succeeding leaves.
For $p\in [n]$ let $b^p=(b^p_{k}\dots b^p_1)$ be the $k$-bit string corresponding to the number $p-1$.
For $i\in\leaves(T)$ we inductively define the $j$th parent $\parr(i,j)$ of $i$ as $\parr(i,0)=i$ and $\parr(i,j)=\parr(\parr(i,j-1))$.
Furthermore, for $i\in \leaves(T)$ and $1\le j\le k$ we define a precede-function $\preced(i,j)$ as
\[\preced(i,j)=\begin{cases}
        \lc(\parr(i,j)), & \text{if }i\in \leaves(T(\rc(\parr(i,j)))) \\
        \rc(\parr(i,j)), & \text{otherwise.}
    \end{cases}\]
The values $\parr(i,j)$ and $\preced(i,j)$ can be precomputed in $\mathcal{O}(n\log n)$ time and stored using $\mathcal{O}(n\log n)$ space.
The value $\rob(p)=\min\{i:1\le i\le k,b^p_i=1\}$ is the \emph{rightmost one bit} of the binary number corresponding to $p-1$ that is 1.

The key insight given by Brandes is stated in the following lemma.
\begin{lem}[\hspace{-0.01pt}\cite{brandesOptimalLeafOrdering2007}]
    If leaf $i\in [n]$ is fixed at position $p\in [n]\setminus \{1\}$ then exactly $\leaves(T(\preced(i,\rob(p))))$ can precede $i$.
\end{lem}
\looseness=-1
Essentially, if a leaf $i$ is fixed at a specific position $p$ of $\pi\in \Pi(T)$, we know exactly which leaves $j\in \leaves(T)$ can precede it in $\pi$. Further, these leaves $j$ exactly correspond to $\leaves(T(\preced(i,\rob(p))))$.
This leads to the algorithm for \cref{thm:nsquard}.
\begin{proof}[Proof of \cref{thm:nsquard}]
    \looseness=-1
    As in the algorithm of Brandes \cite{brandesOptimalLeafOrdering2007} we compute the values $\opt(i,p)$ that correspond to the minimal number of blocks of a prefix of length $p$ ending with leaf $i$.
    When computing $\opt(i,p)$ we simply have to look at all values $\opt(j,p-1)$ for leaves $j$ that can precede $i$ when $i$ is fixed at position $p$. But in our case, we know that $d(j,i)$ is 1 for all $j$ with the only exception of $j=i-1$.
    Thus, we can store $\min_{j}\opt(j,p-1)$ at the internal vertex that corresponds to the lowest common ancestor of all these leaves $j$. Additionally, we have to check for the case where $i-1$ can precede $i$.
    Formally, we compute the following.
    For internal nodes $v$ let $\opt(v,p)=\min(\opt(\lc(v),p), \opt(\rc(v), p))$.
    Clearly $\opt(i,1)=1$ for all $i\in [n]$.
    For $i\in \leaves(T)$ and $p>1$ we can compute $\opt(i,p)$ as
    \[\opt(i,p)=\begin{cases}\opt(\preced(i, \rob(p)),p-1)+1,\qquad\text{if }i-1\not\in \leaves(T(\preced(i,\rob(p)))) \\
            \min(\opt(\preced(i, \rob(p)),p-1)+1,\opt(i-1,p-1))\qquad\text{otherwise.}\end{cases}\]
    The case $i-1\in \leaves(T(\preced(i,\rob(p))))$ exactly corresponds to the possibility of $i-1$ preceding $i$.
    This condition can be checked by precomputing $\lca(i-1,i)$ for all $2\le i\le n$.
    Clearly, all values can be computed in $\mathcal{O}(n^2)$ time and
    \[\min_{\pi^\prime\in\Pi(T)}\blocks(\pi^\prime)=\min_{i\in [n]}\opt(i,n).\]
    With the values in the $\opt$-array, it is immediate how to also compute the permutation $\pi\in \Pi(T)$ with $\blocks(\pi)=\min_{i\in [n]}\opt(i,n)$ in $\mathcal{O}(n^2)$ space and time.
    But the space complexity can even be reduced to $\mathcal{O}(n\log n)$ as shown by Brandes \cite{brandesOptimalLeafOrdering2007}.
    We will only give the high-level idea here, the full description can be found in \cite{brandesOptimalLeafOrdering2007}.
    First, when computing the value $\min_{\pi^\prime\in\Pi(T)}\blocks(\pi^\prime)$ we do not need to store the values $\opt(i,p)$ and $\opt(v,p)$ for all $p\in [n]$. We can simply iterate $p$ from $1$ to $n$ and only keep $\opt$-values for $p$ and $p-1$. The main idea is then to only store the element in the middle of the optimal permutation and recursively determine the optimal permutation under this boundary condition in the first and the second half. This somewhat resembles a single-pivot quicksort approach. The time complexity is $T(n)=2T(n/2)+n^2$ which solves to $T(n)=\mathcal{O}(n^2)$. This approach still works for our adaptation as the recursive procedure splits the search space into two equal parts, both still forming a complete binary tree. Thus, the $\opt$-values can again be stored at internal nodes. The space complexity is dominated by $\mathcal{O}(n\log n)$ for storing $\parr(i,j)$ and $\preced(i,j)$.

    Lastly the algorithm can be adapted for breakpoints instead of blocks by setting $\opt(1,1)=0$ in the beginning and setting $\opt(n,n)=\opt(n,n)-1$ at the end of the algorithm. \qed
\end{proof}

\section{FPT-Algorithm}\label{sec:fpt}
In this section we will show an FPT-algorithm for \revision{\bcrprob}\ for arbitrary binary trees parameterized by the number of transpositions.
The proposed algorithm will be able to produce a witness for \revision{\bcrprob}\ in the same time if such a witness exists.

First, let us state two lemmata that are used for reduction rules.
The first lemma shows that by removing an element from a permutation we never increase the transposition distance.
\begin{restatable}{lem}{lemmapermdelete}\label{lemma:permdelete}
    Let $\pi\in\Pi_n$ and $i\in [n]$. Then $d_t(\pi\ominus\pi_i)\le d_t(\pi)$.
\end{restatable}
\begin{proof}
    Let $\tau^1,\dots, \tau^k$ be an $\id$-transposition sequence of $\pi$. For $\ell\in [k]$, let $\upsilon^\ell=\tau^\ell\ominus\inv((\pi\circ \tau^1\circ\dots\circ \tau^{\ell-1}))(\pi_i)$. It is easy to see that all $\upsilon^\ell$ are transpositions, and an exhaustive case distinction in \cref{appendix:equationproof} shows that
    \begin{equation}(\pi\circ\tau^1\circ\dots\circ\tau^\ell)\ominus\pi_i=(\pi\ominus\pi_i)\circ\upsilon^1\circ\dots\circ\upsilon^\ell.\label{eq:transpositiontransformation}
    \end{equation}
    Hence, $(\pi\ominus\pi_i)\circ \upsilon^1\circ\dots \upsilon^k=\id_{n-1}$. Essentially, the $\id$-transposition sequence $\tau^1,\dots,\tau^k$ for $\pi$ is simulated on $\pi\ominus \pi_i$ with the transposition sequence $\upsilon^1,\dots,\upsilon^k$. Hence, $d_t(\pi\ominus\pi_i)\le d_t(\pi)$.

\end{proof}
This can be used to show that, w.r.t.\ transposition distance, it is in some sense always ``safe'' to shift an element $x\in\pi$ after $x-1$ or before $x+1$ in a permutation.
Let us give the formal definition and lemma.
\begin{restatable}{defin}{defadjbefaft}\label{def:adjbefaft}
    Let $\pi\in \Pi_n$, $i\in [n]$, and $\pi_i=x$.
    Define $\adjbef(\pi, x)$ as the permutation obtained from $\pi$ by shifting $x$ to the beginning if $x=1$, otherwise to the position after $x-1$. Formally,
    \begin{equation*}\adjbef(\pi, x)=\\
        \begin{cases}
            (x,\pi_1,\dots,\pi_{i-1},\pi_{i+1},\dots,\pi_n)\qquad\qquad\qquad\qquad\qquad \text{if }x=1                                  \\
            (\pi_1,\dots,\pi_{j-1},x-1,x,\pi_{j+1},\dots,\pi_{i-1},\pi_{i+1},\dots,\pi_n) \\
            \qquad\qquad\qquad\qquad\qquad\qquad\text{if }x\ne 1\text{ and }i>j=\inv(\pi)(x-1)  \\
            (\pi_1,\dots,\pi_{i-1},\pi_{i+1},\dots\pi_{j-1},x-1,x,\pi_{j+1},\dots,,\pi_n) \\
            \qquad\qquad\qquad\qquad\qquad\qquad\text{if }x\ne 1\text{ and }i<j=\inv(\pi)(x-1).
        \end{cases}
        \end{equation*}
        
    Similarly, define $\adjaft(\pi, x)$ as the permutation obtained from $\pi$ by shifting $x$ to the end if $x=n$, otherwise to the position before $x+1$. Formally,
    \begin{equation*}
        \adjaft(\pi, x)=\begin{cases}
            (\pi_1,\dots,\pi_{i-1},\pi_{i+1},\dots,\pi_n,x)\qquad\qquad\qquad\qquad\qquad\text{if }x=n                                  \\
            (\pi_1,\dots,\pi_{j-1},x,x+1,\pi_{j+1},\dots,\pi_{i-1},\pi_{i+1},\dots,\pi_n)\\
            \qquad\qquad\qquad\qquad\qquad\qquad\text{if }x\ne n\text{ and }i>j=\inv(\pi)(x+1)  \\
            (\pi_1,\dots,\pi_{i-1},\pi_{i+1},\dots\pi_{j-1},x,x+1,\pi_{j+1},\dots,,\pi_n)\\
            \qquad\qquad\qquad\qquad\qquad\qquad \text{if }x\ne n\text{ and }i<j=\inv(\pi)(x+1).
        \end{cases}
    \end{equation*}
\end{restatable}
\begin{restatable}{lem}{lemmapermshiftto}\label{lemma:permshiftto}
    Let $\pi\in \Pi_n$, $i\in [n]$, and $\pi_i=x$.
    Then $d_t(\adjbef(\pi, x))\le d_t(\pi)$ and $d_t(\adjaft(\pi, x))\le d_t(\pi)$.
\end{restatable}
\begin{proof}
    We show that $d_t(\adjbef(\pi, x))\le d_t(\pi)$, $d_t(\adjaft(\pi, x))\le d_t(\pi)$ can be shown in an equivalent way.
        If $\adjbef(\pi, x)=\pi$, the statement is trivial.
        So, assume $\adjbef(\pi, x)\ne\pi$. In this case we have that $\gl(\adjbef(\pi, x))=\gl(\pi\ominus x)$.
        Further, $d_t(\adjbef(\pi, x))=d_t(\gl(\adjbef(\pi, x)))=d_t(\gl(\pi\ominus x))$ and $d_t(\pi)=d_t(\gl(\pi))$ by \cref{lemma:gluetransdist}.
        Hence, the statement follows from \cref{lemma:permdelete}. \qed
    \end{proof}
Intuitively, the algorithm in Theorem~\ref{thm:bcrfpt} at the end of this section first exhaustively applies a reduction rule based on \cref{lemma:permshiftto} and removes leaves for which we know how they will be ordered in the tree in an optimal solution. Then, it applies a branching scheme that branches into the two orders of the children of an inner node in a bottom-up fashion, such that we can bound the depth of the resulting search tree. During this branching scheme, another reduction rule similar to the first one is applied, and in the leaf nodes of the search tree an FPT-algorithm for \transsort\ is applied.

We start with the first reduction rule (see \cref{fig:redrule1} for an example) based on \cref{lemma:permshiftto} that allows us to delete a pair of children from the input tree $T$ which are siblings and whose difference is one.
\begin{figure}[t]
    \centering
    \begin{subfigure}[b]{0.4\textwidth}
        \centering
        \includegraphics[width=\textwidth]{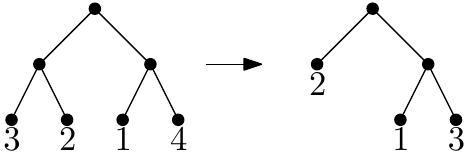}
        \caption{Example of \cref{redrule:1}.}
        \label{fig:redrule1}
    \end{subfigure}
    \hfill
    \begin{subfigure}[b]{0.59\textwidth}
        \centering
        \includegraphics[width=.8\linewidth]{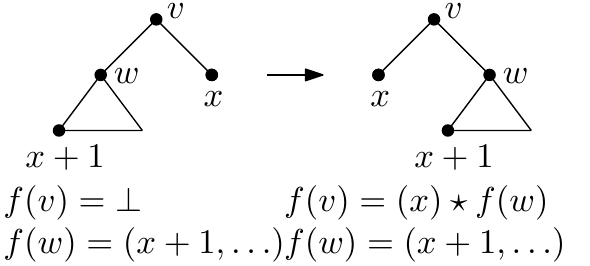}
        \caption{Illustration of \cref{redrule:2}.}
        \label{fig:redrule2}
    \end{subfigure}
    \caption{The reduction rules for the FPT-algorithm.}
    \label{fig:redrules}
\end{figure}
\begin{redrule}\label{redrule:1}
    If there are two leaves $x,y\in \leaves(T)$ which are siblings such that $x+1=y$,
    (1) delete the leaf nodes corresponding to $x$ and $y$ from $T$, (2) replace their parent node, which is now a leaf, with $x$, and (3) for each leaf $z\in\leaves(T)$ with $z>y$, set $z=z-1$.
\end{redrule}
If $\leaves(T)=[n]$, then after the application of \cref{redrule:1} we have that $\leaves(T')=[n-1]$ for the new tree $T'$. The safeness of \cref{redrule:1} follows directly from \cref{lemma:permshiftto} and \cref{lemma:gluetransdist}. Further, if we obtain a witness permutation $\pi'\in\Pi(T')$ with $d_t(\pi')\le k$, this permutation can easily be transformed into a permutation $\pi\in\Pi(T)$ with $d_t(\pi)\le k$ by increasing elements in the permutation that are larger than $x$ by one and replacing the element $x$ with the pair $(x,x+1)$. An $\id$-transposition sequence for $\pi'$ can then be mimicked in $\pi$.

\looseness=-1
After exhaustively applying \cref{redrule:1} we employ a search-tree algorithm such that in each search tree node some part of the ordering corresponding to sets $\leaves(T(v))$ for internal nodes $v$ is already fixed.
Formally, in each search tree node we are given a function $f$ whose domain is $V(T)$ such that either $f(v)=\bot$, or $f(v)\in \Pi(T(v))$. Further, if $f(v)\in \Pi(T(v))$ and $v$ is not a leaf node, then we have for the children $w,x$ of $v$ that $f(w)\ne\bot$ and $f(x)\ne\bot$, and $f(v)$ is \emph{consistent with} $f(w)$ and $f(x)$, that is, $f(v)=f(w)\star f(x)$ or $f(v)=f(x)\star f(w)$.
A permutation $\pi\in \Pi(T)$ is \emph{consistent with} $f$ if for all $v\in V(T)$ with $f(v)\ne \bot$, $\pi=\phi\star f(v)\star \psi$ for some $\phi,\psi$ which can also be the empty permutation.
We call the problems $(T,f)$ associated with the search tree nodes \revision{\bcrextension}.
The question is if there is a permutation $\pi\in \Pi(T)$ consistent with $f$ such that $d_t(\pi)\le k$.
In the root of our search tree and in the initial problem equivalent to \revision{\bcrprob}, $f(x)=(x)$ for all leaves $x\in\leaves(T)$ and $f(v)=\bot$ for all inner nodes.

A key observation is that $f$ already tells us something about the number of breakpoint in any permutation $\pi\in\Pi(T)$ consistent with $f$.
We thus let $\breakpoints(f)$ be the number of pairs $x,y\in\leaves(T)$ that will form a breakpoint in any permutation $\pi\in\Pi(T)$ consistent with $f$. Notice that $\breakpoints(f)$ can be determined in polynomial time as the breakpoint pairs $(x,y)$ s.t.\ $x+1\ne y$ correspond to pairs of adjacent elements in some permutations $f(v)\ne \bot$.

Let us now present the next reduction rule that is applied during the search tree algorithm (see \cref{fig:redrule2} for an illustration), and the branching rule that is employed once the reduction rule is not applicable.
\begin{redrule}\label{redrule:2}
    Let $v\in V(T)$ be an inner node with the two children $x$ and $w$, such that $f(v)=\bot$, $x$ is a leaf, $f(w)\ne\bot$, and $f(w)(1)=x+1$ or $f(w)(|T(w)|)=x-1$.
    If $f(w)(1)=x+1$, set $f(v)=(x)\star f(w)$.
    Otherwise, set $f(v)=f(w)\star (x)$.
\end{redrule}
The safeness of \cref{redrule:2} again follows from \cref{lemma:permshiftto}, as it essentially shifts $x$ before $x+1$ or after $x-1$ in any permutation consistent with $f$. Notice that if $f(w)(1)=x+1$ and $f(w)(|T(w)|)=x-1$ it does not matter how we order the children of $v$, as resulting permutations $\pi\in\Pi(T)$ consistent with $f$ will be equivalent w.r.t.\ the glue-operation.
After exhaustive application of \cref{redrule:2}, the main branching rule can be applied.
\begin{branchrule}\label{branchrule:1}
    Let $v\in V(T)$ be an inner node with the two children $u$ and $w$ such that $f(v)=\bot$, $f(u)\ne\bot$, and $f(w)\ne\bot$. Create two new branches $(T,f^{\alpha})$ and $(T,f^\beta)$ where
    \[f^\alpha(x)=\begin{cases}
            f(u)\star f(w) & \quad\text{if }x=v,    \\
            f(x)           & \quad\text{otherwise},
        \end{cases}\qquad
        f^\beta(x)=\begin{cases}
            f(w)\star f(u) & \quad\text{if }x=v,    \\
            f(x)           & \quad\text{otherwise}.
        \end{cases}\]
\end{branchrule}
\cref{branchrule:1} essentially tries ordering the two children of an inner node in the two possible ways, if the corresponding orderings of subtrees rooted at the children are already determined by $f$.
If $u$ and $w$ are both leaves, then
\begin{equation}\breakpoints(f^\alpha)=\breakpoints(f^\beta)=\breakpoints(f)+1,\label{eq:bpincreases}\end{equation}
as the pair $(u,w)$ or $(w,u)$ will contribute one breakpoint. This is the case as \cref{redrule:1} was already applied exhaustively.
Also, if only one of $u$ and $w$ is a leaf, and \cref{redrule:2} was already applied exhaustively, then \cref{eq:bpincreases} also holds:
If, e.g., $u$ is a leaf, then $(u,f(w)(1))$ or $(f(w)(|\leaves(T(w))|), u)$ will be a breakpoint, depending on the chosen order of $u$ and $w$.
This insight will allow us to bound the depth of the search tree.

With \cref{branchrule:1} we are now ready to give the theorem that captures the algorithm.
\begin{restatable}{theorem}{thmbcrfpt}\label{thm:bcrfpt}
    \revision{\bcrprob}\ is solvable in time $\mathcal{O}(2^{6k}\cdot (3k)^{3k}\cdot n^{\mathcal{O}(1)})$ for rooted binary trees, i.e., \revision{\bcrprob}\ is FPT for rooted binary trees when parameterized by the number of transpositions (= block crossings) $k$.
\end{restatable}
\begin{algorithm}[!t]
    \SetKwFunction{FMain}{Recursive-\revision{\bcrextension}}
    \SetKwProg{Fn}{}{:}{}
    \Fn{\FMain{$T$, $f$}}{
        Exhaustively apply \cref{redrule:2} to $(T,f)$\;
        \If{$\breakpoints(f)>3k$}{\label{line:breakpoints}
            \Return \False
        }
        \If{$f(\rroot(T))\ne \bot$}{
            $\pi\gets f(\rroot(T))$\;
            \uIf(\tcp*[h]{$\mathcal{O}(n(3k)^{3k})$ FPT algorithm}){$d_t(\pi)\le k$}{
                    \Return $\pi$ and $\id$-transposition sequence $(\tau^1,\dots,\tau^\ell)$ with $\ell\le k$
                }
            \uElse{
                \Return \False
            }
        }
        \For(\tcp*[h]{$(T',f')$ runs over the two branches created by applying \cref{branchrule:1} to $T$}){$(T^\prime,f^\prime)\in\textup{\cref{branchrule:1}}(T)$}{\label{line:branchrule}
            $\textup{ans}\gets \FMain(T^\prime,f^\prime)$\;
            \If{$\textup{ans}\ne\False$}{
                \Return ans\;
            }
        }
        \Return \False\;
    }
    \caption{Recursive FPT-algorithm for \revision{\bcrextension}.}
    \label{alg:fptbcrprob}
\end{algorithm}
\begin{proof}
    We give a search-tree algorithm for \revision{\bcrprob}\ that also provides a witness in case of success.
    This search tree algorithm assumes that \cref{redrule:1} was already applied exhaustively and is given as a recursive function in \cref{alg:fptbcrprob}.
    \cref{redrule:1} can be implemented in linear time, and as already mentioned, a witness for the instance after application of the reduction rule can be transformed into a witness before application of the reduction rule.
    Further, the instance $T$ of \revision{\bcrprob}\ is transformed into an instance $(T,f)$ of \revision{\bcrextension}\ by setting $f(x)=(x)$ for all leaves $x\in\leaves(T)$ and $f(v)=\bot$ for all internal nodes.
    The algorithm is then invoked with \FMain$(T,f)$, and returns \False\ in case of failure.
    In case of success it returns $\pi\in\Pi(T)$ and an $\id$-transposition sequence $\tau^1,\dots,\tau^\ell$ with $\ell\le k$ for $\pi$.

    \proofparagraph{Correctness.}
    In each search tree node (function call of \cref{alg:fptbcrprob}), \cref{redrule:2} is applied exhaustively first. The safeness of this rule was already discussed and given because of \cref{lemma:permshiftto}.
    Then, if $\breakpoints(f)>3k$ we can safely disregard the current branch of the search tree because in resulting recursion calls the number of breakpoints can only increase, and any $\id$-transposition sequence will be longer than $k$ (\cref{lemma:lowerbound}).

    If then $f(\rroot(T))\ne \bot$ we know that the order amongst all children of $T$ is fixed. We then apply the algorithm for \transsort\ outlined in \cref{par:transpositions} that tries to find a transposition sequence $\tau^1,\dots,\tau^\ell$ turning $f(\rroot(T))$ into the identity permutation with $\ell\le k$ transpositions. If the algorithm succeeds then we can return $f(\rroot(T))$ and the transposition sequence, otherwise we report failure for this leaf of the search tree.

    In the remaining case we apply \cref{branchrule:1} to create new recursion calls for the two possible orders of children for some internal node. Note that \cref{branchrule:1} is always applicable in \cref{line:branchrule} as \cref{redrule:2} was already applied exhaustively and the order amongst children of at least one internal node is yet to be determined, e.g., $\rroot(T)$.

    \proofparagraph{Runtime.} Note that the search tree is a binary tree because \cref{branchrule:1} always creates two child nodes. Further, we give a bound for the depth of the search tree as follows. Consider any leaf instance $(T,f)$ of the search tree. Because of \cref{line:breakpoints} we have $\breakpoints(f)\le 3k$.
    Let $V_b(T,f)$ be the set of internal nodes of $T$ for which \cref{branchrule:1} fixed the order amongst children of these internal nodes following the path in the search tree to the leaf instance.
    Hence, $|V_b(T,f)|$ is exactly the depth of the leaf instance $(T,f)$ in the search tree.
    Now let $V_{b}^{(2)}(T,f)\subseteq V_b(T,f)$ be the internal nodes whose children are both leaves, and let $V_{b}^{(1)}(T,f)\subseteq V_b(T,f)$ be the internal nodes which have exactly one child that is a leaf.
    In other words, for all $v\in V_{b}^{(2)}(T,f)$ we have $\child(v)\subseteq\leaves(T)$ and $f(v)\ne\bot$. For all $v\in V_{b}^{(1)}(T,f)$ we have $|\child(v)\cap\leaves(T)|=1$ and $f(v)\ne\bot$.
    Let $V_b^i(T,f)=V_b(T,f)\setminus(V_{b}^{(1)}(T,f)\cup V_{b}^{(2)}(T,f))$.
    Because \cref{redrule:1} was already applied exhaustively, each node in $V_{b}^{(2)}(T,f)$ increases $\breakpoints(f)$ by one.
    Also, as we only apply \cref{branchrule:1} once \cref{redrule:2} is not applicable anymore, each node in $V_{b}^{(1)}(T,f)$ increases $\breakpoints(f)$ by one.
    Hence, we have $|V_{b}^{(2)}(T,f)|+|V_{b}^{(1)}(T,f)|\le 3k$.
    Now let us build a forest $F_b$ on the vertex set $V(F_b):=V_b^i(T,f)\cup V_b^{(2)}(T,f)$ as follows.
    For each vertex $v\in V(F_b)$ let $\ances_{F_b}(v)=\mathsf{anc}_T(v)\cap V(F_b)$.
    If $\ances_{F_b}(v)\ne\emptyset$, let $p_v=\argmax_{w\in \ances_{F_b}(v)}\depth_T(w)$ and add the edge $\{v,p_v\}$ to $F_b$.
    \begin{figure}[t]
        \centering
        \includegraphics{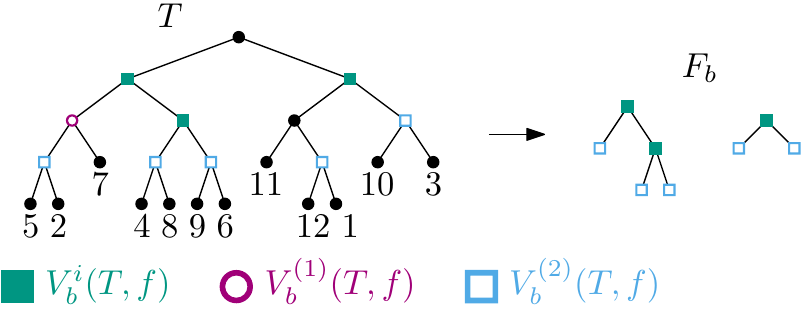}
        \caption{Example for the construction of the forest $F_b$, where \cref{redrule:2} was applied for the parent of $11$, and $f(\rroot(T))=\bot$.}
        \label{fig:fpttimecomplexity}
    \end{figure}
    An example for the construction of $F_b$ is given in \cref{fig:fpttimecomplexity}.
    Now, notice that $F_b$ consists of a set of disjoint rooted binary trees, whose combined leaf set is exactly $V_{b}^{(2)}(T,f)$.
    Hence, $|V(F_b)|\le 2|V_{b}^{(2)}(T,f)|$, and $|V_b(T,f)|\le 2|V_{b}^{(2)}(T,f)|+|V_{b}^{(1)}(T,f)|$. As $|V_{b}^{(2)}(T,f)|+|V_{b}^{(1)}(T,f)|\le 3k$ we can conclude that $|V_b(T,f)|\le 6k$, which bounds the search tree depth.
    It follows that the search tree has at most $\mathcal{O}(2^{6k})$ leaf nodes.
    For each of these leaf nodes, an algorithm for \transsort\ is applied that takes at most $\mathcal{O}(n(3k)^{3k})$ time (see \cref{par:transpositions}).
    Finally, we can conclude that the algorithm can be implemented in time $\mathcal{O}(2^{6k}(3k)^{3k}n^c)$ time for some small constant $c$. \qed
\end{proof}

While we are giving an FPT algorithm here, we think that the asymptotic running time can be further optimized.
Nonetheless, the bottleneck of the algorithm is still an FPT subprocedure for \transsort. So any better algorithm for \transsort\ will also improve Theorem~\ref{thm:bcrfpt}.

We also believe that a similar algorithm can be constructed for non-binary trees. In that case, the algorithm should have running time $\mathcal{O}(f(k,\Delta)\cdot n^c)$, where $\Delta$ is the maximum degree of the input tree $T$, $c$ is a constant, and $f$ is some function that only depends on $k$ and $\Delta$.

\section{Beyond Binary Trees}\label{sec:beyond}
\looseness=-1
In \cref{sec:approx} we have given two approximation algorithms for \revision{\bcrprob}\ on binary trees. The key step was to find in polynomial time a permutation consistent with the input tree that minimizes the number of breakpoints. In the following we show that this is not possible for non-binary trees. We start by showing that finding a permutation with  less than or equal $k$ blocks is \NP-complete.
\begin{restatable}{theorem}{thmblockprobnp}\label{thm:blockprobnp}
    For a rooted tree $T$ with leaf set $[n]$ \revision{and an integer $k$} it is \NP-complete to decide whether there exists $\pi\in\Pi(T)$ s.t.\ $\blocks(\pi)\le k$.
\end{restatable}
\begin{figure}[!t]
    \centering
    \includegraphics[width=.9\linewidth]{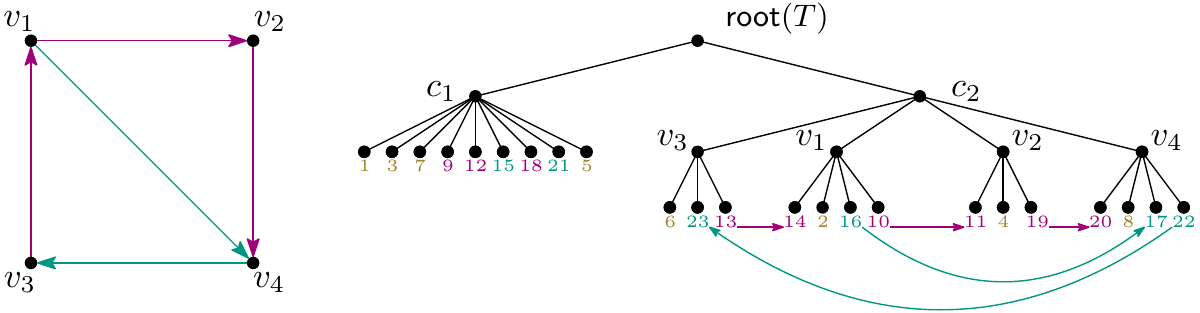}
    \caption{Example of reduction in Theorem~\ref{thm:blockprobnp}.}
    \label{fig:blockprobnp}
\end{figure}
Clearly, the problem is in NP.
For hardness, we give a reduction from the Hamiltonian Path problem on directed graphs with at most one arc between a pair of vertices which is \NP-complete \cite{DBLP:books/fm/GareyJ79}.
Let $G=(V,E)$ be an instance of Hamiltonian Path with directed arcs $E$. The problem is to find a path $P$ in $G$ that visits every vertex exactly once, called \emph{Hamiltonian path}.
Let $V=\{v_1,v_2,\dots,v_{|V|}\}$ and $E=\{e_1,e_2,\dots,e_{|E|}\}$.
W.l.o.g.\ we assume that every vertex in $G$ is incident to at least one edge.
We construct a rooted tree $T$ such that $\leaves(T)=[3|E|+2|V|]$.
An illustration is given in \cref{fig:blockprobnp}.
The root of $T$ has two children $c_1$ and $c_2$.
Vertex $c_1$ contains $|V|+|E|$ children
\[\revision{\{1,3,\dots, 2|V|-1\}\cup \{2|V|+1,2|V|+4,2|V|+7,\dots,2|V|+3|E|-2\}.}\]
Vertex $c_2$ contains $|V|$ children corresponding to the vertex set $V$. Let $\delta_G(v)$ be the \emph{degree} of a vertex $v$ in $G$.
In $T$, vertex $v_i$ has $1+\delta_G(v_i)$ children which are the following leaves: one child is $2i$; for each edge $e_j$ incident to $v_i$ one child is $2|V|+3(j-1)+2$ if $v_j$ is the source of $e_j$, and $2|V|+3(j-1)+3$ otherwise.
Notice that $T$ contains $2|V|+3|E|$ leaves and that $\leaves(T)=[2|V|+3|E|]$, as intended.
The intuition is that each edge $e\in E(G)$ corresponds to a triple $(\ell,\ell+1,\ell+2)$ of leaves in $T$ such that $\ell$ is the child of $c_1$, and $\ell+1$, $\ell+2$ are children of vertices corresponding to the source and target of $e$. The leaf $\ell$ rules out the possibility of a block of size greater than one amongst the children of a vertex $v_i$ in $T$.
The leaves $1,2,\dots,2|V|$ make it possible that there is a block formed by the rightmost leaf in of $T(c_1)$ and the leftmost leaf in $T(c_2)$.
This allows us to show that $G$ contains a Hamiltonian path if and only if there exists $\pi\in\Pi(T)$ with $\blocks(\pi)\le 2|V|+3|E|-1-(|V|-1)$.
\begin{lem}\label{claim:nphardblockforward}
    If $G$ contains a Hamiltonian path then there exists $\pi\in\Pi(T)$ with $\blocks(\pi)\le 2|V|+3|E|-1-(|V|-1)$.
\end{lem}
\begin{proof}
    Let $(v_{i_1},v_{i_2},\dots,v_{i_{|V|}})$ be a Hamiltonian path in $G$. We describe the permutation $\pi\in\Pi(T)$ by giving the order amongst children of every internal node of $T$
    \begin{itemize}
        \item The children of $\rroot(T)$ are ordered such that $c_1$ comes before $c_2$.
        \item The children of $c_1$ are ordered such that $2i_1-1$ is rightmost, the rest is ordered arbitrarily.
        \item The children of $c_2$ are ordered according to the Hamiltonian path. That is, $v_{i_1}$ is before $v_{i_2}$, $v_{i_2}$ is before $v_{i_3}$, and so on.
        \item For a vertex $v_{i_k}$ its children are ordered as follows.
              If $k=1$ then the leftmost child of $v_{i_k}$ is $2i_k$. Otherwise, let $e_j$ be the edge corresponding to the directed arc $(v_{i_{k-1}}, v_{i_k})$. Then, the leftmost child of $v_{i_k}$ is $2|V|+3(j-1)+3$ and the rightmost child of $v_{i_{k-1}}$ is $2|V|+3(j-1)+2$.
    \end{itemize}
    By construction, we have that $\pi\in\Pi(T)$. It remains to show that $\blocks(\pi)=2|V|+3|E|-1-(|V|-1)$. Notice that $\pi$ can have at most $2|V|+3|E|$ blocks and that each block of size two reduces this amount by one. Hence, we show that there are $1+(|V|-1)$ blocks of size two.
    The first of these blocks is obtained by the rightmost leaf in $T(c_1)$ and the leftmost leaf in $T(c_2)$. The remaining $|V|-1$ of these blocks are obtained by the rightmost child of $v_{i_k}$ and the leftmost child of $v_{i_{k+1}}$. These appear consecutively in $\pi$, correspond to the directed arc $(v_{i_k},v_{i_{k+1}})$ in $G$, and hence form a block. \qed
\end{proof}
\begin{lem}\label{claim:nphardblockbackward}
    If there exists $\pi\in\Pi(T)$ with $\blocks(\pi)\le 2|V|+3|E|-1-(|V|-1)$ then $G$ contains a Hamiltonian path.
\end{lem}
\begin{proof}
    Let $\pi\in\Pi(T)$ with $\blocks(\pi)\le 2|V|+3|E|-1-(|V|-1)$. First notice that $\pi$ cannot have any blocks of size greater than two by construction: Every pair of consecutive children in $\child(T(c_1))$ forms a breakpoint as there are no children which form consecutive numbers, and as discussed above, there are no non-breakpoints \revision{amongst two children of a single vertex $v_i$}. Hence, $\pi$ contains $1+(|V|-1)$ blocks of size two. The children of $c_1$ either all appear left or all appear right of $\leaves(T(c_2))$.
    Thus, the blocks of size two can be described as follows. At most one is \revision{formed by} a child of $c_1$ and a child of $v_i$ with $i\in [|V|]$. The remaining ones \revision{are formed by} two leaves $x$, $y$ with $x+1=y$ and such that $x$ is the child of some $v_{i_\alpha}$ and $y$ is the child of some $v_{i_\beta}$. Notice that by construction such a block can only exist if there is a directed arc from $v_{i_\alpha}$ to $v_{i_\beta}$ in $G$. As there are $|V|-1$ adjacencies of that kind, we define the Hamiltonian path $P$ such that $v_{i_\alpha}$ is the predecessor of $v_{i_\beta}$. \qed
\end{proof}
\begin{proof}[Proof of Theorem~\ref{thm:blockprobnp}]
    Notice that the given reduction is polynomial as $T$ has exactly $2|V|+3|E|$ leaves. The correctness follows from \cref{claim:nphardblockforward} and \cref{claim:nphardblockbackward}. \qed
\end{proof}

A simple reduction from the above problem shows the following corollary.
\begin{restatable}{corr}{corrbpnpcomplete}\label{corr:bpnpcomplete}
    For a rooted tree $T$ with leaves $[n]$ it is \NP-complete to decide whether there exists $\pi\in\Pi(T)$ s.t.\ $\breakpoints(\pi)\le k$.
\end{restatable}
\begin{figure}[t]
    \centering
    \includegraphics{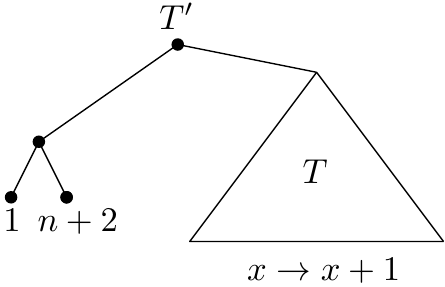}
    \caption{Construction of the tree $T^\prime$ in \cref{corr:bpnpcomplete}.}
    \label{fig:blocksbreakpoints}
\end{figure}
\begin{proof}
    Let $T$ be a tree with $\leaves(T)=[n]$. Construct the tree $T'$ as follows (see \cref{fig:blocksbreakpoints} for a sketch):
    \begin{itemize}
        \item Set $x$ to $x+1$ for every leaf in $T$.
        \item Create a new root with connections to $\rroot(T)$ and a new vertex $u$.
        \item Add leaves $1$ and $n+2$, and connect them to $u$.
    \end{itemize}
    Notice that there exists $\pi\in\Pi(T)$ with $\blocks(\pi)\le k$ if and only if there exists $\pi^\prime\in\Pi(T^\prime)$ with $\breakpoints(\pi^\prime)\le k+2$. Hence, the statement follows from Theorem~\ref{thm:blockprobnp}. \qed
\end{proof}
We believe though, that both problems become tractable if we fix an upper bound on the maximum degree of the tree or the number of blocks/breakpoints.

Hence, the techniques applied in \cref{sec:approx} to obtain a polynomial time approximation algorithm do not extend to non-binary trees. This does not imply, however, that there is no polynomial time approximation for \revision{\bcrprob}.

\section{Summary and Open Problems}
We have analyzed the complexity of minimizing block crossings in one-sided (binary) tanglegrams from different theoretical perspectives. A number of open problems and interesting research directions remain:
\begin{itemize}
    \item We have considered one tree to be fixed. What happens if we can permute the leaves of both trees?
    \item \transsort\ admits better approximations than those we have presented here. However, these algorithms rely on a more complicated structure called \emph{breakpoint graph} (see, e.g., \cite{DBLP:journals/siamdm/BafnaP98,DBLP:journals/tcbb/EliasH06,DBLP:journals/iandc/HartmanS06,DBLP:journals/almob/SilvaKRW22}). Can these approximations be utilized for tanglegrams?
    \item A pair of edges could cross multiple times in the tanglegrams produced by our algorithms. Can this be prevented using a similar notion such as \emph{monotone block crossings} (see \cite{DBLP:journals/jgaa/FinkPW15})?
    \item How do our algorithms perform in practice? What is the relation between number of pairwise crossings and number block crossings in practice?
\end{itemize}

\bibliographystyle{splncs04}
\bibliography{literature.bib}

\begin{thebibliography}{10}
\providecommand{\url}[1]{\texttt{#1}}
\providecommand{\urlprefix}{URL }
\providecommand{\doi}[1]{https://doi.org/#1}

\bibitem{afp-bcn-16}
Alam, M.J., Fink, M., Pupyrev, S.: The bundled crossing number. In: Hu, Y.,
  N{\"{o}}llenburg, M. (eds.) Proc. 24th Symposium on Graph Drawing and Network
  Visualization (GD). LNCS, vol.~9801, pp. 399--412. Springer (2016).
  \doi{10.1007/978-3-319-50106-2\_31}

\bibitem{DBLP:journals/siamdm/BafnaP98}
Bafna, V., Pevzner, P.A.: Sorting by transpositions. {SIAM} J. Discret. Math.
  \textbf{11}(2),  224--240 (1998). \doi{10.1137/S089548019528280X}

\bibitem{barjosephKaryClusteringOptimal2003}
Bar{-}Joseph, Z., Demaine, E.D., Gifford, D.K., Srebro, N., Hamel, A.M.,
  Jaakkola, T.S.: K-ary clustering with optimal leaf ordering for gene
  expression data. Bioinform.  \textbf{19}(9),  1070--1078 (2003).
  \doi{10.1093/bioinformatics/btg030}

\bibitem{DBLP:conf/wea/BaumannBL10}
Baumann, F., Buchheim, C., Liers, F.: Exact bipartite crossing minimization
  under tree constraints. In: Festa, P. (ed.) Proc. 9th Symposium on
  Experimental Algorithms (SEA). LNCS, vol.~6049, pp. 118--128. Springer
  (2010). \doi{10.1007/978-3-642-13193-6\_11}

\bibitem{bockerFasterFixedParameterApproach2009}
B{\"{o}}cker, S., H{\"{u}}ffner, F., Tru{\ss}, A., Wahlstr{\"{o}}m, M.: A
  faster fixed-parameter approach to drawing binary tanglegrams. In: Chen, J.,
  Fomin, F.V. (eds.) Proc. 4th Workshop on Parameterized and Exact Computation
  (IWPEC). LNCS, vol.~5917, pp. 38--49. Springer (2009).
  \doi{10.1007/978-3-642-11269-0\_3}

\bibitem{brandesOptimalLeafOrdering2007}
Brandes, U.: Optimal leaf ordering of complete binary trees. J. Discrete
  Algorithms  \textbf{5}(3),  546--552 (2007). \doi{10.1016/j.jda.2006.09.003}

\bibitem{buchinDrawingCompleteBinary2012}
Buchin, K., Buchin, M., Byrka, J., N{\"{o}}llenburg, M., Okamoto, Y., Silveira,
  R.I., Wolff, A.: Drawing (complete) binary tanglegrams - hardness,
  approximation, fixed-parameter tractability. Algorithmica  \textbf{62}(1-2),
  309--332 (2012). \doi{10.1007/s00453-010-9456-3}

\bibitem{bulteauSortingTranspositionsDifficult2012}
Bulteau, L., Fertin, G., Rusu, I.: Sorting by transpositions is difficult.
  {SIAM} J. Discret. Math.  \textbf{26}(3),  1148--1180 (2012).
  \doi{10.1137/110851390}

\bibitem{DBLP:conf/cpm/BulteauGS22}
Bulteau, L., Gambette, P., Seminck, O.: Reordering a tree according to an order
  on its leaves. In: Bannai, H., Holub, J. (eds.) Proc. 33rd Symposium on
  Combinatorial Pattern Matching (CPM). LIPIcs, vol.~223, pp. 24:1--24:15
  (2022). \doi{10.4230/LIPIcs.CPM.2022.24}

\bibitem{christieGenomeRearrangementProblems1998}
Christie, D.A.: Genome Rearrangement Problems. Ph.D. thesis, University of
  Glasgow (1998), \url{https://theses.gla.ac.uk/74685/}

\bibitem{DBLP:journals/jgaa/DijkFFLMRSW17}
van Dijk, T.C., Fink, M., Fischer, N., Lipp, F., Markfelder, P., Ravsky, A.,
  Suri, S., Wolff, A.: Block crossings in storyline visualizations. J. Graph
  Algorithms Appl.  \textbf{21}(5),  873--913 (2017). \doi{10.7155/jgaa.00443}

\bibitem{ds-olohdptv-04}
Dwyer, T., Schreiber, F.: Optimal leaf ordering for two and a half dimensional
  phylogenetic tree visualisation. In: Churcher, N., Churcher, C. (eds.)
  Australasian Symposium on Information Visualisation (InVis.au). {CRPIT},
  vol.~35, pp. 109--115. Australian Computer Society (2004),
  \url{http://crpit.scem.westernsydney.edu.au/abstracts/CRPITV35Dwyer.html}

\bibitem{DBLP:journals/tcbb/EliasH06}
Elias, I., Hartman, T.: A 1.375-approximation algorithm for sorting by
  transpositions. {IEEE} {ACM} Trans. Comput. Biol. Bioinform.  \textbf{3}(4),
  369--379 (2006). \doi{10.1109/TCBB.2006.44}

\bibitem{DBLP:journals/dm/ErikssonEKSW01}
Eriksson, H., Eriksson, K., Karlander, J., Svensson, L.J., W{\"{a}}stlund, J.:
  Sorting a bridge hand. Discret. Math.  \textbf{241}(1-3),  289--300 (2001).
  \doi{10.1016/S0012-365X(01)00150-9}

\bibitem{fernauComparingTreesCrossing2010}
Fernau, H., Kaufmann, M., Poths, M.: Comparing trees via crossing minimization.
  J. Comput. Syst. Sci.  \textbf{76}(7),  593--608 (2010).
  \doi{10.1016/j.jcss.2009.10.014}

\bibitem{fhsv-bceg-16}
Fink, M., Hershberger, J., Suri, S., Verbeek, K.: Bundled crossings in embedded
  graphs. In: Kranakis, E., Navarro, G., Chávez, E. (eds.) Proc. 12th
  Symposium on Theoretical Informatics (LATIN). LNCS, vol.~9644, pp. 454--468.
  Springer (2016). \doi{10.1007/978-3-662-49529-2\_34}

\bibitem{DBLP:journals/jgaa/FinkPW15}
Fink, M., Pupyrev, S., Wolff, A.: Ordering metro lines by block crossings. J.
  Graph Algorithms Appl.  \textbf{19}(1),  111--153 (2015).
  \doi{10.7155/jgaa.00351}

\bibitem{DBLP:books/fm/GareyJ79}
Garey, M.R., Johnson, D.S.: Computers and Intractability: {A} Guide to the
  Theory of NP-Completeness. W. H. Freeman (1979)

\bibitem{DBLP:journals/iandc/HartmanS06}
Hartman, T., Shamir, R.: A simpler and faster 1.5-approximation algorithm for
  sorting by transpositions. Inf. Comput.  \textbf{204}(2),  275--290 (2006).
  \doi{10.1016/j.ic.2005.09.002}

\bibitem{hw-vchod-08}
Holten, D., van Wijk, J.J.: Visual comparison of hierarchically organized data.
  Comput. Graph. Forum  \textbf{27}(3),  759--766 (2008).
  \doi{10.1111/j.1467-8659.2008.01205.x}

\bibitem{mahajanApproximateBlockSorting2006}
Mahajan, M., Rama, R., Raman, V., Vijaykumar, S.: Approximate block sorting.
  Int. J. Found. Comput. Sci.  \textbf{17}(2),  337--356 (2006).
  \doi{10.1142/S0129054106003863}

\bibitem{DBLP:books/sp/20/Nollenburg20}
N{\"{o}}llenburg, M.: Crossing layout in non-planar graph drawings. In: Hong,
  S., Tokuyama, T. (eds.) Beyond Planar Graphs, Communications of {NII} Shonan
  Meetings, pp. 187--209. Springer (2020). \doi{10.1007/978-981-15-6533-5\_11}

\bibitem{nollenburgDrawingBinaryTanglegrams2009}
N{\"{o}}llenburg, M., V{\"{o}}lker, M., Wolff, A., Holten, D.: Drawing binary
  tanglegrams: An experimental evaluation. In: Finocchi, I., Hershberger, J.
  (eds.) Proc. 11th Workshop on Algorithm Engineering and Experiments (ALENEX).
  pp. 106--119. {SIAM} (2009). \doi{10.1137/1.9781611972894.11}

\bibitem{pageTangledTreesPhylogeny2003}
Page, R.D.M.: Tangled Trees: Phylogeny, Cospeciation, and Coevolution.
  University of Chicago Press (2003)

\bibitem{szh-trptn-11}
Scornavacca, C., Zickmann, F., Huson, D.H.: Tanglegrams for rooted phylogenetic
  trees and networks. Bioinform.  \textbf{27}(13),  248--256 (2011).
  \doi{10.1093/bioinformatics/btr210}

\bibitem{DBLP:journals/almob/SilvaKRW22}
Silva, L.A.G., Kowada, L.A.B., Rocco, N.R., Walter, M.E.M.T.: A new
  1.375-approximation algorithm for sorting by transpositions. Algorithms Mol.
  Biol.  \textbf{17}(1), ~1 (2022). \doi{10.1186/s13015-022-00205-z}

\bibitem{vajg-utcttd-10}
Venkatachalam, B., Apple, J., John, K.S., Gusfield, D.: Untangling tanglegrams:
  Comparing trees by their drawings. {IEEE} {ACM} Trans. Comput. Biol.
  Bioinform.  \textbf{7}(4),  588--597 (2010). \doi{10.1109/TCBB.2010.57}

\bibitem{DBLP:conf/spire/WalterDM00}
Walter, M.E.T., Dias, Z., Meidanis, J.: A new approach for approximating the
  tranposition distance. In: de~la Fuente, P. (ed.) Proc. 7th Symposium on
  String Processing and Information Retrieval (SPIRE). pp. 199--208. {IEEE}
  Computer Society (2000). \doi{10.1109/SPIRE.2000.878196}

\end{thebibliography}

\clearpage
\appendix
\section{\texorpdfstring{Proof of Equation~\eqref{eq:transpositiontransformation} \cref{lemma:permdelete}}{}}\label{appendix:equationproof}
Let us give the proof of \cref{eq:transpositiontransformation}.
    First notice that our $\ominus$-operation can be defined as follows.
    Let $\pi\in\Pi_n$ and $x\in [n]$, then $\pi\ominus x\in\Pi_{n-1}$ is the permutation defined as
    \begin{equation}(\pi\ominus x)(j)=
        \begin{cases}
            \pi(j)     & \quad\text{ if }j<\inv(\pi)(x)\text{ and }\pi(j)<x                         \\
            \pi(j)-1   & \quad\text{ if }j<\inv(\pi)(x)\text{ and }\pi(j)>x                         \\
            \pi(j+1)   & \quad\text{ if }j\ge \inv(\pi)(x)\text{ and }\pi(j+1)<x                    \\
            \pi(j+1)-1 & \quad\text{ if }j\ge \inv(\pi)(x)\text{ and }\pi(j+1)>x\label{eq:minusdef}
        \end{cases}
    \end{equation}
    
    We show \cref{eq:transpositiontransformation} by induction on $\ell$. The base case ($\ell=0$) is trivial.
    Let us continue with the induction step, thus $\ell\ge 1$ is arbitrary.
    Let us consider the right side of \cref{eq:transpositiontransformation}.
    First, by the induction hypothesis, we have that
    \begin{equation}
        (\pi\ominus\pi_i)\circ\upsilon^1\circ\dots\circ\upsilon^\ell=((\pi\circ\tau^1\circ\dots\circ\tau^{\ell-1})\ominus \pi_i)\circ\upsilon^\ell.
    \end{equation}
    Remember that $\upsilon^\ell=\tau^\ell\ominus\inv(\pi\circ \tau^1\circ\dots\circ \tau^{\ell-1})(\pi_i)$.
    Let $\phi=\pi\circ \tau^1\circ\dots\circ \tau^{\ell-1}$.
    Plugging this into \cref{eq:transpositiontransformation}, we get
    \begin{equation}(\phi\circ\tau^{\ell})\ominus \pi_i=(\phi\ominus \pi_i)\circ (\tau^\ell\ominus \inv(\phi)(\pi_i))\label{eq:reducedtranstransformation}
    \end{equation}
    We consider all cases of the left and right side of \cref{eq:reducedtranstransformation} and show that they result in the same permutation.
    Let us first consider the left side.
    \begin{align}\begin{split}
        &((\phi\circ\tau^{\ell})\ominus \pi_i)(j)=\\
        &\begin{cases}
            (\phi\circ\tau^{\ell})(j)     & \text{if }j<\inv(\phi\circ\tau^{\ell})(\pi_i)\text{ and }(\phi\circ\tau^{\ell})(j)<\pi_i                         \\
            (\phi\circ\tau^{\ell})(j)-1   & \text{if }j<\inv(\phi\circ\tau^{\ell})(\pi_i)\text{ and }(\phi\circ\tau^{\ell})(j)>\pi_i                         \\
            (\phi\circ\tau^{\ell})(j+1)   & \text{if }j\ge\inv(\phi\circ\tau^{\ell})(\pi_i)\text{ and }(\phi\circ\tau^{\ell})(j+1)<\pi_i                     \\
            (\phi\circ\tau^{\ell})(j+1)-1 & \text{if }j\ge\inv(\phi\circ\tau^{\ell})(\pi_i)\text{ and }(\phi\circ\tau^{\ell})(j+1)>\pi_i \label{eq:leftside}
        \end{cases}
    \end{split}\end{align}
    Let us now consider the right side of \cref{eq:reducedtranstransformation}. We will have to consider 16 cases, but some of them will be contradicting as highlighted in red.
    \begin{align}\begin{split}
        &((\phi\ominus \pi_i)\circ (\tau^\ell\ominus \inv(\phi)(\pi_i)))(j)=\\
        &\begin{cases}
            (\phi\ominus\pi_i)(\tau^{\ell}(j))   & \text{ if }j<\inv(\tau^{\ell})(\inv(\phi)(\pi_i))\text{ and }\tau^\ell(j)<\inv(\phi)(\pi_i)      \\
            (\phi\ominus\pi_i)
            (\tau^\ell(j)-1)                     & \text{ if }j<\inv(\tau^{\ell})(\inv(\phi)(\pi_i))\text{ and }\tau^\ell(j)>\inv(\phi)(\pi_i)      \\
            (\phi\ominus\pi_i)(\tau^\ell(j+1))   & \text{ if }j\ge \inv(\tau^{\ell})(\inv(\phi)(\pi_i))\text{ and }\tau^\ell(j+1)<\inv(\phi)(\pi_i) \\
            (\phi\ominus\pi_i)(\tau^\ell(j+1)-1) & \text{ if }j\ge \inv(\tau^{\ell})(\inv(\phi)(\pi_i))\text{ and }\tau^\ell(j+1)>\inv(\phi)(\pi_i)
        \end{cases}
    \end{split}\end{align}
    We consider the 4 cases of this equation in more detail in the next 4 paragraphs.
    \paragraph*{Case 1 $j<\inv(\tau^{\ell})(\inv(\phi)(\pi_i))$ and \textcolor{red}{$\tau^\ell(j)<\inv(\phi)(\pi_i)$}.}
    \begin{align}\begin{split}
        &(\phi\ominus\pi_i)(\tau^{\ell}(j))=\\
        &\begin{cases}
            \phi(\tau^\ell(j))                 & \text{ if }\tau^\ell(j)<\inv(\phi)(\pi_i)\text{ and }\phi(\tau^\ell(j))<\pi_i \\
            \phi(\tau^\ell(j))-1               & \text{ if }\tau^\ell(j)<\inv(\phi)(\pi_i)\text{ and }\phi(\tau^\ell(j))>\pi_i \\
            \text{\textcolor{red}{impossible}} & \textcolor{red}{\text{ if }\tau^\ell(j)\ge\inv(\phi)(\pi_i)}
        \end{cases}
    \end{split}\end{align}
    The last case is impossible as $\tau^\ell(j)\ge\inv(\phi)(\pi_i)$ but $\tau^\ell(j)<\inv(\phi)(\pi_i)$ by case 1.
    Otherwise, we have
    \begin{itemize}
        \item $\phi(\tau^\ell(j))$, if $j<\inv(\tau^{\ell})(\inv(\phi)(\pi_i))$ which is equivalent to $j<\inv(\phi\circ\tau^{\ell})(\pi_i)$, and $\phi(\tau^\ell(j))<\pi_i$.
        \item $\phi(\tau^\ell(j))$, if $j<\inv(\tau^{\ell})(\inv(\phi)(\pi_i))$ which is equivalent to $j<\inv(\phi\circ\tau^{\ell})(\pi_i)$, and $\phi(\tau^\ell(j))>\pi_i$.
    \end{itemize}
    If we compare with \cref{eq:leftside}, we notice that the result of the permutations are equivalent for case 1. So \cref{eq:reducedtranstransformation} holds for this case.

    \paragraph*{Case 2 {$j<\inv(\tau^{\ell})(\inv(\phi)(\pi_i))$ and \textcolor{red}{$\tau^\ell(j)>\inv(\phi)(\pi_i)$}}.}
    \begin{align}\begin{split}
        &(\phi\ominus\pi_i)
        (\tau^\ell(j)-1)=\\
        &\begin{cases}
            \textcolor{red}{\text{impossible}} & \textcolor{red}{\text{ if }\tau^\ell(j)-1<\inv(\phi)(\pi_i)}                       \\
            \phi(\tau^\ell(j))                 & \text{ if }\tau^\ell(j)-1\ge \inv(\phi)(\pi_i)\text{ and }\phi(\tau^\ell(j))<\pi_i \\
            \phi(\tau^\ell(j))-1               & \text{ if }\tau^\ell(j)-1\ge \inv(\phi)(\pi_i)\text{ and }\phi(\tau^\ell(j))>\pi_i
        \end{cases}
    \end{split}\end{align}
    The first case is impossible as $\tau^\ell(j)-1<\inv(\phi)(\pi_i)\iff \tau^\ell(j)\le \inv(\phi)(\pi_i)$ but $\tau^\ell(j)>\inv(\phi)(\pi_i)$ by case 2.
    Otherwise, we have
    \begin{itemize}
        \item $\phi(\tau^\ell(j))$, if $j<\inv(\tau^{\ell})(\inv(\phi)(\pi_i))$ which is equivalent to $j<\inv(\phi\circ\tau^{\ell})(\pi_i)$, and $\phi(\tau^\ell(j))<\pi_i$.
        \item $\phi(\tau^\ell(j))$, if $j<\inv(\tau^{\ell})(\inv(\phi)(\pi_i))$ which is equivalent to $j<\inv(\phi\circ\tau^{\ell})(\pi_i)$, and $\phi(\tau^\ell(j))>\pi_i$.
    \end{itemize}
    If we compare with \cref{eq:leftside}, we notice that the result of the permutations are equivalent for case 2. So \cref{eq:reducedtranstransformation} holds for this case.

    \paragraph*{Case 3 $j\ge \inv(\tau^{\ell})(\inv(\phi)(\pi_i))$ and \textcolor{red}{$\tau^\ell(j+1)<\inv(\phi)(\pi_i)$}.}
    \begin{align}\begin{split}
        &(\phi\ominus\pi_i)(\tau^\ell(j+1))=\\
        &\begin{cases}
            \phi(\tau^\ell(j+1))               & \text{ if }\tau^\ell(j+1)<\inv(\phi)(\pi_i)\text{ and }\phi(\tau^\ell(j+1))<\pi_i \\
            \phi(\tau^\ell(j+1))-1             & \text{ if }\tau^\ell(j+1)<\inv(\phi)(\pi_i)\text{ and }\phi(\tau^\ell(j+1))>\pi_i \\
            \textcolor{red}{\text{impossible}} & \textcolor{red}{\text{ if }\tau^\ell(j+1)\ge\inv(\phi)(\pi_i)}
        \end{cases}
    \end{split}\end{align}
    The last case is impossible as $\tau^\ell(j+1)\ge\inv(\phi)(\pi_i)$ but $\tau^\ell(j+1)<\inv(\phi)(\pi_i)$ by case 3.
    Otherwise, we have
    \begin{itemize}
        \item $\phi(\tau^\ell(j+1))$, if $j\ge\inv(\tau^{\ell})(\inv(\phi)(\pi_i))$ which is equivalent to $j\ge\inv(\phi\circ\tau^{\ell})(\pi_i)$, and $\phi(\tau^\ell(j+1))<\pi_i$.
        \item $\phi(\tau^\ell(j+1))-1$, if $j\ge\inv(\tau^{\ell})(\inv(\phi)(\pi_i))$ which is equivalent to $j\ge\inv(\phi\circ\tau^{\ell})(\pi_i)$, and $\phi(\tau^\ell(j+1))>\pi_i$.
    \end{itemize}
    If we compare with \cref{eq:leftside}, we notice that the result of the permutations are equivalent for case 3. So \cref{eq:reducedtranstransformation} holds for this case.

    \paragraph*{Case 4 $j\ge \inv(\tau^{\ell})(\inv(\phi)(\pi_i))$ and \textcolor{red}{$\tau^\ell(j+1)>\inv(\phi)(\pi_i)$}.}
    \begin{align}\begin{split}
        &(\phi\ominus\pi_i)(\tau^\ell(j+1)-1)=\\
        &\begin{cases}
            \textcolor{red}{\text{impossible}} & \textcolor{red}{\text{ if }\tau^{\ell}(j+1)-1<\inv(\phi)(\pi_i)}                       \\
            \phi(\tau^{\ell}(j+1))             & \text{ if }\tau^\ell(j+1)-1\ge \inv(\phi)(\pi_i)\text{ and }\phi(\tau^\ell(j+1))<\pi_i \\
            \phi(\tau^\ell(j+1))-1             & \text{ if }\tau^\ell(j+1)-1\ge \inv(\phi)(\pi_i)\text{ and }\phi(\tau^\ell(j+1))>\pi_i
        \end{cases}
    \end{split}\end{align}
    The first case is impossible as $\tau^{\ell}(j+1)-1<\inv(\phi)(\pi_i)\iff \tau^{\ell}(j+1)\le \inv(\phi)(\pi_i)$ but $\tau^\ell(j+1)>\inv(\phi)(\pi_i)$ by case 4.
    Otherwise, we have
    \begin{itemize}
        \item $\phi(\tau^\ell(j+1))$, if $j\ge\inv(\tau^{\ell})(\inv(\phi)(\pi_i))$ which is equivalent to $j\ge\inv(\phi\circ\tau^{\ell})(\pi_i)$, and $\phi(\tau^\ell(j+1))<\pi_i$.
        \item $\phi(\tau^\ell(j+1))-1$, if $j\ge\inv(\tau^{\ell})(\inv(\phi)(\pi_i))$ which is equivalent to $j\ge\inv(\phi\circ\tau^{\ell})(\pi_i)$, and $\phi(\tau^\ell(j+1))>\pi_i$.
    \end{itemize}
    If we compare with \cref{eq:leftside}, we notice that the result of the permutations are equivalent for case 4. So \cref{eq:reducedtranstransformation} holds for this case.

    In all cases \cref{eq:reducedtranstransformation} holds, which completes the induction step. \qed

\end{document}